\documentclass[letterpaper,11pt]{article}

\usepackage[hmargin=1in, vmargin=1in]{geometry}

\usepackage{parskip}
\usepackage{amsmath}
\usepackage{amssymb}
\usepackage{xspace}
\usepackage{float}

\usepackage{algorithm}
\usepackage[noend]{algpseudocode}

\usepackage{booktabs}

\usepackage{smartref} 
\usepackage{comment}
\excludecomment{ignore}

\newtheorem{xtheorem}{Theorem}
\newtheorem{xdefinition}[xtheorem]{Definition}
\newtheorem{xobservation}[xtheorem]{Observation}
\newtheorem{xlemma}[xtheorem]{Lemma}
\newtheorem{xproposition}[xtheorem]{Proposition}
\newtheorem{xcorollary}[xtheorem]{Corollary}
{\hspace*{\fill}\raisebox{-1pt}{\boldmath$\Box$}\end{xdefinition}}
\newenvironment{observation}{\begin{xobservation}\rm}%
{\hspace*{\fill}\raisebox{-1pt}{\boldmath$\Box$}\end{xobservation}}
\newenvironment{theorem}{\begin{xtheorem}\rm}{\end{xtheorem}}
\newenvironment{lemma}{\begin{xlemma}\rm}{\end{xlemma}}

\newenvironment{proof}{\begin{trivlist}\item[]{\bf Proof }}%
{\hspace*{\fill}\raisebox{-1pt}{\boldmath$\Box$}\end{trivlist}}

\algrenewcommand\algorithmicrequire{\textbf{Precondition:}}
\algrenewcommand\algorithmicensure{\textbf{Postcondition:}}
\newcommand*\Let[2]{\State #1 $\gets$ #2}

\newcommand{\OPT}{\ensuremath{\operatorname{\textsc{OPT}}}\xspace}
\newcommand{\ALG}{\ensuremath{\operatorname{\textsc{Alg}}}\xspace}

\newcommand{\SIZE}[1]{\left|#1\right|}
\newcommand{\SET}[1]{\left\{#1\right\}}
\newcommand{\SETOF}[2]{\SET{#1 \mid #2}}

\newcommand{\MTF}{\ensuremath{\operatorname{\textsc{MTF}}}\xspace}
\newcommand{\DMTF}{\ensuremath{\operatorname{\textsc{DMTF}}}\xspace}

\newcommand{\REV}[1]{\ensuremath{\operatorname{\textsc{rev}}(#1)}\xspace}
\newcommand{\REP}[2]{\ensuremath{#2 \times (#1)}\xspace}
\newcommand{\RS}[2]{\ensuremath{\operatorname{#1}^{(#2)}}\xspace}

\newcommand{\BLOCK}[2]{\ensuremath{\operatorname{\textrm{Block}}(#1,#2)}\xspace}
\newcommand{\PREV}[2]{\ensuremath{\operatorname{\textsc{prev}}^{#1}(#2)}\xspace}
\newcommand{\SUCC}[2]{\ensuremath{\operatorname{\textsc{succ}}^{#1}(#2)}\xspace}
\newcommand{\REL}[2]{\ensuremath{\operatorname{\textit{d}}^{#1}(#2)}\xspace}
\newcommand{\TOTAL}[1]{\ensuremath{\operatorname{\textit{d}}(#1)}\xspace}
\newcommand{\ITEMSUM}[2]{\ensuremath{\operatorname{\textrm{item}}^{#1}(#2)}\xspace}

\newcommand{\NEXTOT}[1]{\ensuremath{\operatorname{\textrm{S}}^{I\leftarrow J}_{#1}}\xspace}
\newcommand{\NEXTTO}[1]{\ensuremath{\operatorname{\textrm{S}}^{J\leftarrow I}_{#1}}\xspace}
\newcommand{\NEXTGEN}[3]{\ensuremath{\operatorname{\textrm{S}}^{#2\leftarrow #3}_{#1}}\xspace}

\newcommand{\PP}{\ensuremath{\operatorname{\mathbb{P}}}\xspace}
\newcommand{\PART}[1]{\ensuremath{\operatorname{G}^{#1}}\xspace}
\newcommand{\PRIMEPART}[1]{\ensuremath{\operatorname{H}^{#1}}\xspace}
\newcommand{\Mod}[1]{\ \mathrm{mod}\ #1}
\newcommand{\Ip}{\sigma_1,\sigma_2,\ldots,\sigma_p}
\newcommand{\BL}[1]{\ensuremath{\operatorname{\textsc{Lin}}_{\epsilon}(#1)}\xspace}
\newcommand{\INDEX}{\ensuremath{\operatorname{\textrm{cnt}}}\xspace}
\newcommand{\Dproblem}{\ensuremath{\operatorname{\textsc{FindValue}}}\xspace}

\newcounter{linenum}
\addtoreflist{linenum}
\def\codeTabSpace{\hspace*{6mm}}
\newenvironment{code}%
{\begin{tabbing}%
\codeTabSpace \= \hspace*{30mm} \= \hspace*{40mm} \= \hspace*{42mm} \= \kill%
}%
{\end{tabbing}%
}
\newcounter{ind}
\newcommand{\n}{\addtocounter{ind}{7}\hspace*{7mm}}
\newcommand{\p}{\addtocounter{ind}{-7}\hspace*{-7mm}}
\newcommand{\nl}{\\\stepcounter{linenum}{\scriptsize \arabic{linenum}}\>\hspace*{\value{ind}mm}}
\newcommand{\ul}{\\\>\hspace*{\value{ind}mm}}

\newcommand{\lref}[1]{\linenumref{#1}} 

\newcommand{\NULL}{\ensuremath{\operatorname{\textsc{null}}}\xspace}
\newcommand{\DONE}{\ensuremath{\operatorname{\textsc{done}}}\xspace}
\newcommand{\GONE}{\ensuremath{\operatorname{\textsc{gone}}}\xspace}

\title{The Scheduler is Very Powerful in Competitive Analysis of Distributed
List Accessing\,
\thanks{Supported
in part by the Independent Research Fund Denmark, Natural Sciences,
grant DFF-7014-00041 and the Natural Science and Engineering Research Council of Canada.}}

\author{
        Joan Boyar\,$\mbox{}^{1}$ \hspace{2em}
        Faith Ellen\,$\mbox{}^{2}$ \hspace{2em}
        Kim S. Larsen\,$\mbox{}^{1}$ \\[2ex]
        $\mbox{}^1$
        University of Southern Denmark,
        {\tt \{joan,kslarsen\}@imada.sdu.dk} \\[1ex]
        $\mbox{}^2$
        University of Toronto,
        {\tt faith@cs.toronto.edu}
}      

\date{}

\begin{document}

\maketitle

\begin{abstract}
This work is a continuation of efforts to define and understand
competitive analysis of algorithms in a distributed shared memory 
setting,
which is surprisingly different from the classical online setting. In
fact, in a distributed shared memory setting, we find a counter-example
to the theorem concerning classical randomized online algorithms which shows
that, if there is a $c$-competitive randomized algorithm against an
adaptive offline adversary, then there is a $c$-competitive deterministic
algorithm~\cite{BBKTW94}. In a distributed setting, there is additional
lack of knowledge concerning what the other processes have done. There
is also additional power for the adversary, having control of the scheduler
which decides when each process is allowed to take steps.

We consider the list accessing problem, which is a benchmark
problem for sequential online algorithms. In the distributed version of this problem,
each process has its own finite sequence of requests to a shared list.
The scheduler arises as a major issue in its competitive analysis.
We introduce two different adversaries,
which differ in how they are allowed to schedule processes,
and use them to perform competitive analysis of distributed
list accessing.
We prove tight upper and lower bounds on combinatorial properties of
merges of the request sequences,
which we use in the analysis.
Our analysis shows that the effects of the adversarial scheduler
can be quite significant, dominating the
usual quality loss due to lack of information about the future.
\end{abstract}

\section{Introduction}
Our aim is to improve our understanding of 
competitive analysis of algorithms in a distributed shared memory setting.
We investigate a theoretical benchmark online problem,
List Accessing. 
New problems arise due to scheduling, when the processes each
have their own finite request sequences. These problems are not addressed in
existing models for competitive analysis in distributed settings, so we
introduce two new adversarial models.
Through a sequence of results, we expose which
circumstances affect worst-case performance the most.
In standard online algorithms, the lack of knowledge of the
future is often the primary obstruction, but in the
distributed setting, the scheduling of events seems to play
an even larger role.

\subsection{Online Algorithms and Competitive Analysis}

A problem is called \emph{online} if the input is given one piece at a time and an algorithm
must make an irrevocable decision regarding each piece before the next piece is given.
The pieces of the input are called \emph{requests}. An \emph{online algorithm} is an algorithm
for solving an online problem.  The goal of an online algorithm is to minimize or maximize
an objective function. In this paper,
we restrict the discussion to minimization problems and call the objective function the \emph{cost}.

The most common technique for analyzing online algorithms is competitive analysis~\cite{ST85}.
If $\ALG(I)$ denotes the cost of running the algorithm
\ALG on the input sequence~$I$, then \ALG is said to be
\emph{$c$-competitive}, if there exists a constant $\alpha$ such that
for all $I$, $\ALG(I)\leq c\OPT(I) + \alpha$,
where \OPT is a (hypothetical) offline, optimal algorithm.
The \emph{competitive ratio} of an algorithm is the infimum over all
$c$ for which it is $c$-competitive.

\subsection{Competitive Analysis for Distributed Algorithms}
We consider a standard model of asynchronous, shared
memory computation~\cite{AW04}
in which $p$ processes communicate with one another by
performing read, write, and compare-and-swap (CAS) on shared memory locations.
CAS($x,old,new$) atomically returns the value of $x$ and, if it is $old$, changes it to $new$.
In this model, processes may run at arbitrarily varying speeds and may even crash.
Each shared memory access is one step, and the goal, after correctness,
is to minimize the number of steps.
An execution is modeled as a sequence of input events, output events, and shared memory
accesses starting from an initial configuration.
The schedule (that is, the order in which processes take steps)
and what inputs they are given are viewed as being under the control of an adversary.

Overall, the purpose of applying competitive analysis (comparing to some optimal algorithm, \OPT)
is to give a realistic idea of how much improvement of an algorithm might be possible,
as opposed to simply using a worst-case analysis.
Its goal is to  measure the excess cost of a distributed algorithm as compared to the cost of an optimal algorithm in related situations.

In studying competitive analysis for distributed algorithms, we first
consider the well-known theorem from classic online 
algorithms which says that
 if there is a $c$-competitive randomized algorithm against an
adaptive, offline adversary, then there is a $c$-competitive deterministic
algorithm~\cite{BBKTW94}.
We present a counter-example, showing
that this is not true in a distributed setting.

\begin{theorem}
There exists a problem, \Dproblem, in a distributed setting,
where there is a randomized algorithm which is $\frac{23}{16}$-competitive
against an adaptive, offline adversary, while the
best deterministic algorithm is no better than $\frac{3}{2}$-competitive.
\end{theorem}

The first papers applying competitive analysis in a distributed
setting~\cite{AKP92,BFR95} applied it in a message passing, rather
than a shared memory setting. They
compared the cost of a distributed online algorithm to the cost of
an optimal offline, sequential algorithm.
Even so, they achieved good competitive ratios for a job scheduling
problem and a data management problem.
This same definition
of competitive analysis 
is used
in~\cite{AA94,BR97,MMadHVW97,ABF03}
and elsewhere, including~\cite{ABF98} with extra resource analysis.
In~\cite{AKRS94}, a lower bound is proven in this model, explicitly assuming
that the request sequence is given sequentially.
In these papers, \OPT has global control and does not pay for the overhead of
learning the relevant part of the global state to make its decisions.
Aspnes and Waarts~\cite{AW05} have argued
that this might be fair for problems where the purpose is to
manage resources, but unfair for problems where the main
purpose is to propagate information.
The first paper to propose a model of competitive analysis without global control was~\cite{AADW94}. 

Online algorithms are faced with uncertainty as a result of lack of knowledge about future requests:
irrevocable decisions regarding early requests may be unfortunate when future requests arrive.
In his survey on the competitive analysis of distributed
algorithms~\cite{Aspnes-Survey}, Aspnes observed that distributed online algorithms also
face uncertainly about the scheduler.
Actions of the processes may be scheduled in a beneficial
or counter-productive manner with regards to the overall objective function.
Aspnes compares algorithms to \OPT on the same schedule. 
He considers comparing them on the same input,
on the worst-case input for the algorithm and the best-case input for  \OPT, or
on the worst-case input for the algorithm and the worst-case input for  \OPT.
He considers the repeated collect problem in a shared memory model
and presents the proof that the randomized Follow-the-Bodies algorithm~\cite{AH98}
is $O(\log^3 n)$-competitive.

In~\cite{AW05}, the focus is modularity, trying to make competitive analysis
compositional. This necessitates that algorithms and \OPT may be compared
on different inputs. 
Alistarh et al.~\cite{AABGG14} compare an asynchronous algorithm to
\OPT on a worst schedule for \OPT.

Online algorithms have been considered for multi-threaded problems,
including paging~\cite{FS02}, multi-sequence prefetching~\cite{K02},
and multi-threaded metrical task systems~\cite{FSS06}.
In these problems, the input consists of $p$ lists of requests, but requests are treated by one server,
who can choose which request to treat next from among the first outstanding requests in each list.
The server being analyzed is not distributed.

\subsection{The Benchmark Problem: List Accessing}
The (static) List Accessing problem is a theoretical benchmark problem.
In fact, it was one of the original two problems studied with competitive
analysis and the first
with amortized analysis~\cite{ST85}. In addition, it is
the first major example and a means of introducing randomized
online algorithms in the standard textbook~\cite{BE98}.
It has been used to explore performance measure issues
using lookahead~\cite{SM94},
bijective analysis~\cite{AS13},
parameterized analysis~\cite{DEL09},
relative worst-order analysis~\cite{EKL13},
and advice complexity~\cite{BKLL17,BFKLM17},
as well as issues related to self-adjusting data structures.

In this problem, we must administrate a linked list of $\ell$ items.
A request is an item in the list,
and the algorithm must find that item in the list, each time starting from
the front of the list.
When the item has been found, the algorithm may move it to any location closer
to the front of the list.
In the sequential online setting,
the (possible) move must be completed before the
next request is given to the algorithm.
The cost is the total number of item accesses during the search.
Thus, an algorithm for this problem is simply a strategy, detailing when
an item should be moved and where it should be placed.
(Paid exchanges, where an algorithm is charged one for switching the
order of two adjacent items, are also allowed, but we only consider them
indirectly when known results are used.)

According to competitive analysis, the well-known Move-to-Front algorithm (\MTF),
which always moves an accessed item to the front of the list, is an optimal deterministic
algorithm in the sequential setting.
\MTF has (strict) competitive ratio $2-\frac{2}{\ell+1}$~\cite{I91}.
(Referring to personal communication, Irani credits Karp and Raghavan with
the lower bound.)

\subsection{Competitive Analysis of Distributed List Accessing}
For our distributed version of List Accessing,
the list, of fixed length $\ell$,
is shared among $p$ processes, each of which has its own
request sequence. Each process must access every item in its own
request sequence in order (unless it fails and stops), always starting
at the
front
of the list for each search.
The cost of a search is the total number of shared memory accesses during
the search. In the worst case, this is asymptotically
proportional to the total number of items accessed, which, for each request,
is the index of the item, starting from the front of the list.

Our algorithm is a distributed version of \MTF, so
each search for an item is followed by a move-to-front operation, unless
the item is already at the
front
of the list. We are not assuming that \OPT
performs a move-to-front after every request.
Even though \OPT knows all sequences, it must respect the order of each
process's requests sequence: Only after it has finished a request, can it proceed
to the next request of the same process, initiating this search from
the beginning of the current shared list.

If the online algorithm and \OPT can be given different request sequences, as in~\cite{AW05},
it is easy to show too pessimistic a lower bound for List Accessing.
For example, the  algorithm gets requests  to the end of the list, whereas \OPT gets
requests to the beginning of the list. Thus, in our models,
the $p$ processes are given exactly the same request sequences in the online algorithm as in \OPT.

Although \OPT with global control is considered too pessimistic for problems
concerned with propagating information~\cite{AW05},
this is not the concern of the List Accesing problem.
Contention among the processes only hurts the performance in the worst case,
because the same work may be done by multiple processes.
In both of our models, \OPT has global control:
it knows the entire schedule in advance. As with earlier work that considers an optimal
algorithm, 
\OPT, with global control, we also assume that \OPT is sequential.
Theorem~\ref{theorem-item-level} shows that
giving \OPT global control only increases the competitive ratio by a factor of 2, which is
the factor that comes from the standard online competitive analysis of \MTF.
In contrast, with our fully adversarial scheduler defined below,
a factor of $2p^2-p$ comes from the power of the scheduler,
which is assumed to give a best possible schedule to \OPT and a worst
possible schedule to the distributed \MTF algorithm. Consider the
following simple example: Suppose the shared list starts as
$L=[ x_1,x_2,\ldots,x_{\ell} ]$ and the $p$ processes each have
the same request sequence, $\sigma = \langle x_{\ell},x_{\ell-1},\ldots,x_1\rangle$,
repeated $s$ times. If the scheduler arranges that each process makes
all of its requests before the next begins, then each item is at the end
of the list when it is requested and has cost $\ell$. If the scheduler
arranges that the processes take steps in a round robin fashion, then
all $p$ processes can be looking for the same last item in the list at
the same time, each checking if it is the first item in the list before
any checks if it is the second, etc. 
Again, all processes have cost $\ell$ for finding each item.
(The problem associated with all trying to perform
the move-to-front simultaneously is an additional problem, discussed later.)
However, if each process finishes the search and move-to-front for their
item before the next begins, all but the first would only have cost~$1$.

Even if \OPT and the distributed \MTF algorithm are given the same schedule,
there does not seem to be any way to use this information. Since \OPT may complete
some operations in fewer steps than the distributed \MTF algorithm, they could
eventually be searching for different items at the same
point in the schedule
and the lists in which they are searching may be different.
Thus, we do not assume the
same schedule for the distributed algorithm and for \OPT.

\subsection{New Adversarial Models}
\label{linearization_section}
We consider a distributed version of \MTF, which we call \DMTF.
A specific lock-free implementation is presented
in Section~\ref{implementation}.
Most of our competitive
analysis of \DMTF does not depend on the specific implementation.

We assume there are $p$ processes  and
each process has its own input sequence.
We analyze \DMTF using two different models. In both, we assume
\DMTF and \OPT are processing the same $p$ input sequences.

First, we allow \OPT to have a completely different schedule than \DMTF.
We call this the \emph{fully adversarial scheduler} model.
With completely different schedules, one can only be certain that
each process performs its requests in the order specified by its request sequence. However,
requests from other
request sequences can be merged into that sequence in any way.
This method appears to be closest to the worst-case spirit generally
associated with competitive analysis. 
As far as we know, this is the first time the merging of processes'
request sequences has been considered in competitive analysis of
a distributed algorithm, where the merging is not fully or partly
under the control of the algorithm being analyzed.

There are three parts to this analysis.
First, we only consider the effect of merging the request sequences,
requiring operations to be performed sequentially.
Next, we consider processes that can be looking for the same item at the same time.
Then, we take into account arbitrary contention among processes.
The first two parts  are independent of the specific implementation of
distributed \MTF.

\begin{itemize}
\item
Since the adversary can schedule \DMTF and \OPT differently, we
bound
the ratio of costs associated with two
different merges of the
processes'
request sequences.
This corresponds to the competitiveness that
can be experienced in a model where interleaving takes place
at the \emph{operation level}, i.e., a search for an item and the
subsequent move-to-front is carried out before the next search
is initiated. Thus, the execution is sequential after the merging
of the request sequences.  
In this level, the cost of  \DMTF is \MTF's cost on the merged sequence.
\item
Then we consider interleaving at the \emph{item access level}. 
Here, we
take into account
that more than one process may be searching for the same
item at the same time. However, we assume that all such processes
find the item in the same location.
Then the item is moved to the front of the list.
Again, we only consider the costs of searching through
the list, counting cost $i$ for an item in location $i$ of the list.
\item
Finally, we consider an actual algorithm in a fully distributed model.
The  overhead involved in a move-to-front operation in a distributed
setting is considered at this point. One complication is that a process may
move an item to the front of the list after another process has started searching for
the same item. Additional overhead is required to ensure that the second process
does not search all the way to the end of the list without finding the item.
\end{itemize}

In the fully adversarial scheduler model, \OPT may have too much of an
advantage over \DMTF, because
\OPT can arbitrarily merge the sequences of requests given to the processes.
To understand the effect of this, we consider
a second model, which is
more similar
to the situation in the sequential setting, where \OPT performs the same
sequence of operations (searches) as \MTF.
In this model, which we call {\em linearization-based}, \OPT still performs
its operations sequentially, but
is required to perform them in an order which is a linearization of the
execution performed by \DMTF.
This means that, if operation $o_1$ is completed by \DMTF before it begins
operation $o_2$,
then \OPT must perform $o_1$ before $o_2$.

The processes that are successful in satisfying a request do so either by
\begin{enumerate}
\item finding the item and moving it to the front (if it is not already there),
\item finding the item and discovering that another process
has moved it to the front,
or
\item being informed about the item by another process.
\end{enumerate}
In all three cases,
the item is moved to the front by some process and
the concurrent requests to that item are linearized as successive
requests to that item, 
when that item is at the front of the list.
Thus, the ordering of the list during any execution 
of \DMTF is the
same
as it would be if the sequential \MTF
algorithm were run on the 
linearized request sequence.

\subsection{Results}

First, we consider the well-known theorem from classic online algorithms
which says that
competitive analysis of randomized algorithms against
adaptive, offline adversaries is uninteresting,
because these randomized algorithms
are no better than deterministic algorithms.
In Section \ref{section:lowerbound},
we present a counter-example, showing
that this is not true in a distributed setting.

The main technical result is a proof in the fully adversarial scheduler
model, showing that one merge of $p$ sequences can be at most a factor 
$2p^2-p$ more costly for \MTF than another merge. It is accompanied by
a matching lower bound. 
This corresponds to interleaving at the operation level. 
The result implies that,
in the fully adversarial scheduler model, the scheduler is so
powerful that it can ensure that \OPT asymptotically does a factor of
$\Theta(p^2)$ better than any sequential algorithm, even
an optimal offline algorithm that is
forced to run sequentially on this adversarial merge.
Thus, randomization cannot help here,
although it often helps in distributed settings.

In order to prove this result, we prove a property about a 
new distance measure defined on merges of $p$ sequences.
We believe that this combinatorial result
could be of independent interest.

At worst, the adversary can choose the input sequences so that, for two different merges,
the cost of \MTF can differ by a factor of $2p^2-p$. Furthermore, for any merge,
the cost of \MTF is at most a factor of 2 larger than the cost of \OPT.
Thus, for interleaving at the operation level, the ratio of the cost of \MTF to the cost of \OPT is
at most $4p^2-2p$. These results are presented in Section \ref{section:operationLevel}.

We consider 
interleaving at the item access level in Section \ref{sec:itemLevel} and show that it
does not increase the
worst-case ratio.
Indeed, there will sometimes be wasted work when processes search
for the same item,
but not in the worst-case scenarios.
The lower bound on the ratio for interleaving at the operation level carries over
to interleaving at the item access level.

In the linearization-based model, both \OPT and \MTF have the same input sequence
for interleaving at the operation level,
so the classic result from online algorithms gives the ratio $2-\frac{2}{\ell+1}$~\cite{I91}.
For interleaving at the item access level, we prove an upper bound of $p+1$ on the ratio,
as opposed to $4p^2-2p$ for the fully  adversarial scheduler.
This appears in Section~\ref{sec:LinBased}.
A
lower bound of $p$ follows from the example where all
$p$ processes are always searching for the same item at the end of the list.

In Section \ref{implementation},
we present a lock-free implementation of a distributed version of \MTF,
which we call \DMTF.
Its complete analysis in either the fully adversarial model with arbitrary interleaving
or the linearization-based model, 
including the analysis of the move-to-front operation,
increases the cost of a search by a constant factor and an additive $O(p^2)$ term.
Thus, in the fully adversarial model, the upper bound on the competitive ratio is $O(p^2)$
and, in the linearization-based model, it is $O(p)$ provided that \OPT's average cost per 
request is $\Omega(p)$.
The lower bounds from the item access level carry over to 
the level of the actual algorithm.

These results concerning List Accessing are summarized in Table~\ref{table:summary}.
Note that the results from the first two levels are independent of the
actual distributed algorithm being used; this is not considered until
the third level.
Thus, the operation level seems to be the dominant level
for the fully adversarial model, and the item access level seems to be
the dominant level for the linearization-based model.

\begin{table}[ht]
\begin{center}
\begin{tabular}{lcc}
\toprule
& \multicolumn{2}{c}{Model} \\
\cmidrule(lr){2-3}
Level & Fully adversarial & Linearization-based \\
\midrule
Operation & $[2p^2-p,4p^2-2p]$ &  $2-\frac{2}{\ell+1}$ \\[1ex]
Item access &  $[2p^2-p,4p^2-2p]$ & $[p,p+1]$  \\[1ex]
Actual algorithm & $\Theta(p^2)$ & $\Theta(p)$ \\
\bottomrule
\end{tabular}
\end{center}
\caption{Bounds on the ratio of \DMTF to \OPT.}
\label{table:summary}
\end{table}

Thus, when using competitive analysis in a distributed
setting, 
there is no dependency on~$\ell$,
the length of the list, which we assume is much larger than the
number of processes. This length would show up in the analysis of the
worst case amount of work done if competitive analysis was not used.

\section{Competitiveness in a Distributed Setting Is Different}
\label{section:lowerbound}

In his survey~\cite{Aspnes-Survey}, Aspnes
presents a competitive analysis of the randomized Follow-the-Bodies algorithm
against an adaptive offline adversary.
He writes that
it is unknown whether there is a
deterministic algorithm that performs this well.
In the standard sequential online model, such a deterministic algorithm must exist~\cite{BBKTW94}. 
However,
this is not necessarily the case for distributed algorithms. Processes
lack information about the state of the system as a whole and must act
based on their limited knowledge of that state.

The theorem in~\cite{BBKTW94} (and \cite[Theorem 7.3]{BE98})
states that, for any problem, if there is a $c$-competitive randomized
online algorithm against an adaptive, offline adversary, then there exists
a $c$-competitive deterministic online algorithm.
We use the rest of this section to show that,
for a natural problem in a distributed setting,
this theorem does not hold.

We define the setting, show matching upper and lower bounds of $3/2$ on
on the deterministic competitive ratio
and give
a randomized synchronous algorithm which is
$\frac{23}{16}$-competitive against an adaptive, offline adversary.

\subsection{The Model and Problem}
We consider the following problem, \Dproblem, where
there are $3$ processes, $p_0$, $p_1$, and $p_2$ in a synchronous distributed system.
Each process, $p_i$, has one register, $R_i$.
The processes
communicate by writing into and
reading from these single-writer registers.
In each round, each process 
can flip some coins and,
based on its state and on the outcomes of its coin tosses, it can
do nothing, it can write to its
register, or it can read the register
of some other process.

Consider the following problem. From time to time, an adversary gives
a number as input to one of the processes and lets it take a step (in
which it appends a pair consisting of its process id and this number to its register.
In the next round, 
the adversary notifies each of the other processes
that it has produced a new number (by giving a special notification
input), but does not tell them to whom it gave this number.
The goal  is for each process to 
write the entire sequence of
pairs
into their single-writer register.
Our
complexity measure is the number of register reads that are performed
(by the two processes trying to find the number).

Each time the adversary 
notifies a process that it has produced a new number, \OPT will perform one read,
from the register of the process that received the number, and append the new pair 
to its single-writer register.  Since it must do this for two processes,
it must perform two register reads for each input.

\subsection{Deterministic Upper Bound}
The following deterministic algorithm performs $3$ reads per input
item.
Each process $p_i$ maintains a list of the pairs it has learned about and the
total number of notifications it has received from the
adversary.
When the adversary gives
a number as input to  process $p_i$, it appends  a pair consisting of  $i$ and this number to its
list and writes its list to $R_i$.
When told that a new number is available, process $p_i$ reads from
$R_{(i+1) \bmod 3}$.  If there
are more pairs in that register than in its list,
$p_i$ appends the extra pairs
to its list
and writes its list to $R_i$. 
If the length of the list in $R_{(i+1) \bmod 3}$ is smaller than the number of
notifications it has received from the adversary, then, 
in the following round, $p_i$ reads from
register $R_{(i-1) \bmod 3}$, appends the extra
pairs in that register to its list,
and writes its list to $R_i$.

When process $p_k$ directly gets a number as input from the adversary,
process $p_{(k-1) \bmod 3}$ will read this number from $R_k$ in the next round.
However, process $p_{(k+1) \bmod 3}$ will read from $R_{(k + 2) \bmod 3}$ in that round
and then will read from $R_k$ in the following round.
Thus, a total of three reads are performed for each input.
Thus, the algorithm is $\frac{3}{2}$-competitive.

\subsection{Deterministic Lower Bound}
An adversary can force $3$ reads per
input number. 
It will not give any process
an input until all processes know about all of the previous inputs.

For each process $p_i$, let $R_{f_i}$ be the 
first location that $p_i$ will read from when informed that the next new
number is available.  Note that $f_i \not= i$.  Suppose there exist
two processes, $p_i$ and $p_j$, 
such that $f_i = f_j = k$. Without loss of generality, suppose
that $f_k = i$.  Then the adversary gives the new input number to
$p_j$.
Whichever of $p_i$ and $p_k$ goes first performs at least 2 reads, for a
total of $3$ reads.

Otherwise, no two processes read from the same register first. Without
loss of generality, suppose that $f_0 = 1$, $f_1 = 2$, and
$f_2 = 0$. Consider the minimum number of rounds after receiving a notification
input before any one of these processes reads a register. Suppose $p_k$ is a register
that reads in this round. 
If the adversary gives the new input number to $p_{(k -1) \bmod 3}$, then
$p_k$ reads from $R_{(k+1) \bmod 3}$ in this round and does not see the new number. 
Hence it has to perform at least 2 reads.
Process $p_{(k+1) \bmod 3}$ also has to perform at least one read, for a total of 3.
 Thus, no deterministic algorithm
 is better than $\frac{3}{2}$-competitive.
 
\subsection{Randomized Upper Bound}

Consider the following randomized algorithm for this problem.
Each process $p_i$ maintains a list of the pairs it has learned about and the
total number of notifications it has received from the adversary.
When the adversary gives
a number as input to  process $p_i$, $p_i$ appends  a pair consisting of  $i$ and this number to its
list and writes its list to $R_i$.
Whenever process $p_i$ has been told that a new number is
available, it flips two fair coins. Based on the outcome of the first coin, $p_i$ decides
whether to read in this round or delay its next read for 2 rounds.
The second coin is used to choose which of the registers belonging to the other two processes
it will read from. Suppose it chooses to read from register $R_j$.
If there
are more pairs in $R_j$ than in its list, $p_i$ appends the extra pairs in
$R_j$ to its list
and writes its list to $R_i$ in the next round.
If the length of $R_i$ is
smaller than the number of notifications it has received from the
adversary, $p_i$ also reads from the register $R_k$ of the other process,
and appends the extra pairs in $R_k$ to its list.

Suppose that $p_i$ is given
the number as input in round $r$. Let $p_j$ and $p_k$ be the other two processes.

With probability $\frac14$, $p_j$ reads from $R_i$ in round $r+1$ and writes the number to $R_j$ in round~$r+2$. In this case,
\begin{itemize}
\item
$p_j$ performs $1$ read,
\item
$p_k$ performs $1$ read with probability $\frac34$, and 
\item
$p_k$ performs $2$ reads with probability $\frac14$ (if it reads from $R_j$ in round~$r+1$).
\end{itemize}

With probability $\frac14$, $p_j$ reads from $R_i$ in round~$r+3$. In this case,
\begin{itemize}
\item
$p_j$ performs $1$ read,
\item
$p_k$ performs $1$ read with probability $\frac12$, and
\item
$p_k$ performs $2$ reads with probability $\frac12$ (depending on whether $p_k$ chooses to read from $R_i$ or $R_j$ first)
\end{itemize}

With probability $\frac14$, $p_j$ reads from $R_k$ in  round~$r+1$.
In this case,
\begin{itemize}
\item
$p_j$ performs $2$ reads,
\item
$p_k$ performs $1$ read with probability $\frac12$, and
\item
$p_k$ performs $2$ reads with probability $\frac12$ (depending on whether $p_k$ chooses to read from $R_i$ or $R_j$ first).
\end{itemize}

With probability $\frac14$, $p_j$ reads from $R_k$ in  round~$r+3$.
In this case, 
\begin{itemize}
\item
with probability $\frac12$, $p_k$ first reads from $R_j$ and both processes perform $2$ reads,
\item
with probability $\frac14$, $p_k$ reads from $R_i$ in round~$r+ 1$,
so both processes perform $1$ read, and
\item
with probability $\frac14$, $p_k$ first reads from $R_i$ in round~$r+3$,
so $p_k$ performs $1$ read and $p_j$ performs $2$ reads.
\end{itemize}

The expected number of reads for this algorithm is

\[\begin{array}{cl}
& \frac14 \times (3/4 \times 2 + \frac14 \times 3) + \frac14 \times (\frac12 \times 2 + \frac12 \times 3) \\[1ex]
+ & \frac14 \times (\frac12 \times 3 + \frac12 \times 4) + \frac14 \times (\frac12 \times 4 + \frac14 \times 2 + \frac14 \times 3) \\[1.5ex]
= & \frac14 \times (\frac94 + \frac52 + \frac72 + \frac{13}{4}) = \frac{23}{8} < 3.
\end{array}\]

Thus,
there is a randomized algorithm which is $\frac{23}{16}$-competitive against an
adaptive, offline adversary, giving us the following:
\begin{theorem}
There exists a problem, \Dproblem, in a distributed setting,
where there is a randomized algorithm which is $\frac{23}{16}$-competitive
against an adaptive, offline adversary, while the
best deterministic algorithm is no better than $\frac{3}{2}$-competitive.
\end{theorem}
This is a counterexample showing that the theorem in~\cite{BBKTW94}, making results
against adaptive, offline adversaries uninteresting in the sequential online setting, does not
necessarily apply to distributed settings.

\section{\MTF, \OPT, and the Distance Measure}
\label{section:terminology}

A \emph{request sequence} is a sequence of (requests to) items from a list of size $\ell$.
A request to the same item
can appear multiple times in a request sequence, so we are often working with the index
into a sequence.
We use $I$, $J$, and $K$ to refer to generic request sequences.
For some request sequence $I$ of length $n = |I|$, its sequence of 
requests
is denoted
$I_1, I_2, \ldots, I_{n}$. 
Our notation is case-sensitive, so $I_i$ denotes the
request
with index $i$ 
in the request sequence $I$.
We use $i$, $j$, $k$, $x$, $y$, and $z$ to denote indices.
The concatenation of two sequences
$I$ and $J$ is written as $IJ$, and the reverse of a sequence $I$
is written as $\REV{I}$.
If $J$ is a sequence of requests, we use $\REP{J}{s}$ to denote the
concatenation of $s$ copies of list $J$. So, for example,
$\REP{J}{3}$ denotes $J J J$.
In notation referring to a particular request sequence, we
will often indicate the sequence as a superscript.

For any index $j\in\SET{1, \ldots, |I|}$,
let $\PREV{I}{j}$ 
denote the index $j' < j$ of the latest request
in $I$ such that $I_j=I_{j'}$, if it exists.
Thus, there are no requests to that item in
$I_{j'+1}, \ldots, I_{j-1}$.
Similarly, let $\SUCC{I}{j}$ denote the index of the earliest request to $I_j$
after $j$, if it exists.

We define the \emph{distance}, $\REL{I}{j}$, of the $j$th request
in $I$ to be
the cardinality of the set of items in requests from index $\PREV{I}{j}$ to index $j$, i.e.
$$\SIZE{\SET{I_{\PREV{I}{j}}, \ldots, I_j}}=\SIZE{\SET{I_{\PREV{I}{j}}, \ldots, I_{j-1}}},$$
if $\PREV{I}{j}$ exists, and $\ell$ otherwise.
Note that multiple requests to the same item in such a sequence
are only counted once.
Similar definitions of distance have been used 
for Paging~\cite{AF15}, for example.

We extend the notation to sets and sequences of indices so that
$\REL{I}{S} =\sum_{j \in S}\REL{I}{j}$ 
is the sum of distances for all indices in a set or sequence, $S$,
of indices into a request sequence $I$.
We define $\TOTAL{I} = \REL{I}{\SET{1, \ldots, \SIZE{I}}}$ and refer to this as the
\emph{total} distance of $I$.

The distance measure very closely reflects \MTF's cost on a sequence.
In the upper bound proof, it is used
as a tool for measuring the difference in cost incurred by \MTF on two different merges.
\begin{lemma}
\label{lemma-mtf-dist}
For any request sequence $I$,
  \[ \TOTAL{I}-\frac12\ell^2+\frac12\ell \leq \MTF(I) \leq \TOTAL{I}. \]
\end{lemma}
\begin{proof}
Consider a request $I_x$ in $I$ that is not 
the first request to one of the items.
The cost for \MTF to serve $I_x$ is exactly $\REL{I}{x}$,
since it has moved $\REL{I}{x}-1$ different items in front of
the item requested by $I_x$ since it last moved that item to the front of the list.

By definition, the distance of the first request to each item is $\ell$,
which is an upper bound on the  cost for \MTF to serve the request,
since the list has length $\ell$.
Thus, the total distance $\TOTAL{I}$ is an upper bound on the cost for \MTF to serve
all the requests in $I$.

Suppose that $I$ contains requests to $k\leq \ell$ different items.
The item originally in position $j$ in the list will be in position $j$ or
larger when it is requested the first time, so the cost of the
first request to that item is at least~$j$. Thus, the total cost for \MTF
to serve the first requests to all $k$ items is at least
$\sum_{j=1}^k j = \frac12(k^2 + k)$.
Since $\TOTAL{I} -k\ell$ is the cost for all subsequent requests to items,
$\TOTAL{I}-k\ell+\frac12(k^2 + k)  \geq \TOTAL{I}-\frac12\ell^2+\frac12\ell$
is a lower bound on the cost for \MTF to serve
all the requests in $I$.
\end{proof}

Next,
we give a lower bound on the cost of \OPT in terms of the distance measure.
We use the fact that the strict competitive ratio of
\MTF is $2 - \frac{2}{\ell+1}$~\cite{I91}.

\begin{lemma}
\label{lemma-opt-lower-bound}
  For any request sequence $I$ of length~$n$,
  \[ \OPT(I) \geq \frac{\TOTAL{I}}{2} \frac{\ell+1}{\ell} - \frac{\ell^2-1}{4} \]
\end{lemma}
\begin{proof}
  \[\begin{array}{rcl}
  \OPT(I) & \geq & \frac{\MTF(I)}{2 - \frac{2}{\ell+1}},
              \mbox{by the competitiveness of \MTF} \\[2ex]
          & \geq & \frac{\TOTAL{I}-\frac12\ell^2+\frac12\ell}{2 - \frac{2}{\ell+1}},
              \mbox{by Lemma~\ref{lemma-mtf-dist}} \\[2ex]
          & =    & \frac{\TOTAL{I}}{2} \frac{\ell+1}{\ell}  -
                   \frac{\ell^2-1}{4}.
  \end{array}\]
\end{proof}

A \emph{merge} of some request sequences is an interleaving of their 
requests,
respecting the order of each sequence,
as in mergesort.
We often let $M$ denote a merged sequence, and $C$
the particularly simple merge which is the concatenation of
the sequences in order of their process numbers.
The merge will be fixed throughout proofs, so we do not
add it to the notation.

For the fully adversarial scheduler model, we reduce the problem 
of comparing \DMTF to \OPT at the operation level
to considering the ratio of the distances of two merges of the
request sequences.
In the next section, we
present tight bounds on the ratio $\frac{\TOTAL{M_1}}{\TOTAL{M_2}}$,
for two different merges $M_1$ and $M_2$ of the same $p$ request sequences.

The following theorem reduces the competitiveness problem to that
of determining the worst-case distance ratio.
We use that the strict competitive ratio of
\MTF is $2 - \frac{2}{\ell+1}$~\cite{I91}.

\begin{theorem}
\label{theorem-reduced-to-ratio}
When \MTF is processing the sequence $M_1$ and \OPT is processing the
sequence $M_2$, both of the same length, then the ratio of \MTF to \OPT
is at most $(2-\frac{2}{\ell+1})\frac{\TOTAL{M_1}}{\TOTAL{M_2}}$.
\end{theorem}
\begin{proof}
From Lemmas~\ref{lemma-mtf-dist} and~\ref{lemma-opt-lower-bound},
we get a ratio of at most
\[\frac{\TOTAL{M_1}}{\frac{\TOTAL{M_2}}{2}\frac{\ell+1}{\ell}} =
  (2 - \frac{2}{\ell+1})\frac{\TOTAL{M_1}}{\TOTAL{M_2}}.\]
\end{proof}

In the next section, we show an asymptotically tight result, stating
that the ratio,
$\frac{\TOTAL{M_1}}{\TOTAL{M_2}}$, is $2p^2-p$ in the worst case for large $\ell$.

\section{Worst-Case Distance Ratio Between Merges}
\label{section:operationLevel}
In this section,
we find the exact worst-case ratio, up to an additive constant, 
between the total distances of two different merges of the same
$p$ sequences. The results are the main technical contributions of
this paper.

\subsection{Lower Bound}

We show that the maximal ratio
of the average distances for two merges of the same
$p$ sequences of requests to a list of length~$\ell$
is approximately $2p^2-p$ for $\ell$ significantly larger than~$p$.
\begin{theorem}
\label{theorem-ratio-lower-bound}
  The maximum ratio of the average distances for two merges of the same
  $p$ sequences of requests to a list of length~$\ell$ is at least
  $2p^2 - p - \frac{4(p^4 - p^3)}{\ell + 2p^2 -p}$.
\end{theorem}

\begin{proof}
For ease of presentation, we assume that $p$ divides $\ell$.

For the request sequences we construct,
the basic building blocks are
$\RS{A}{j} = (j-1)\frac{\ell}{p}+1, \ldots, j\frac{\ell}{p}$,
for $j\in\SET{1,\ldots,p}$,
each of which is a sequence of requests to $\ell/p$ consecutive items.
Define 
\emph{block} $\RS{B}{j}$, for $j\in\SET{1,\ldots,p}$,
to be $s$ repetitions of the concatenation of $\RS{A}{j}$ and $\REV{\RS{A}{j}}$,
i.e., $\RS{B}{j} = \REP{\RS{A}{j} \REV{\RS{A}{j}}}{s}$.

The request sequence for process~$i$ is now defined to be
\[\sigma_i = \REP{\RS{B}{i} \RS{B}{(i\Mod{p})+1}  \RS{B}{((i+1) \Mod{p})+1}\cdots 
\RS{B}{((i+p-2)\Mod{p})+1}}{r},\]
which consists of $r$ repetitions of requests to all of the $p$ different blocks, in order,
starting with block $\RS{B}{i}$.
Thus, each process's sequence consists of $rp$ blocks.

Let
$\BLOCK{\sigma_i}{h}$ denote the $h$'th block in $\sigma_i$.
Thus, $\BLOCK{\sigma_i}{h}$ is block $\RS{B}{j}$ for some~$j$.

Since each process starts with a different block, it follows that,
for any fixed $h$, each of the $p$ blocks $\BLOCK{\sigma_1}{h}, \ldots, \BLOCK{\sigma_p}{h}$
contains requests to a disjoint set of $\frac{\ell}{p}$ items.
Thus, one possible merge is
\[ \REP{\RS{A}{1} \cdots \RS{A}{p} \REV{\RS{A}{1}} \cdots \REV{\RS{A}{p}}}{prs} \]
where $\RS{A}{1} \cdots \RS{A}{p}$ is $1, \ldots, \ell$.

The average distance of this merge is close to $\ell$.
We now compute the exact value.
Consider
$\RS{A}{1}$ and $\REV{\RS{A}{1}}$, which have the sequence
$\frac{\ell}{p}+1, \ldots, \ell$ in between them in the merge:
\[1, 2, \ldots, \frac{\ell}{p}, \frac{\ell}{p} + 1, \ldots, \ell, \frac{\ell}{p}, \frac{\ell}{p}-1, \ldots, 2, 1\]
Considering the last $\ell/p$ requests,
the distance for the request to item $1$ is $\ell$, the distance for the request to item $2$
is $\ell-1$, and the distance for the request to item $\frac{\ell}{p}$ is $\ell+1-\frac{\ell}{p}$.
Thus, the average distance over the $\frac{\ell}{p}$ items in $\REV{\RS{A}{1}}$ is
$\ell - \frac{1}{\ell/p}\Sigma_{i=1}^{\ell/p}(i-1) = \ell - \frac{\ell}{2p} + \frac12=\frac{(2p-1)\ell+p}{2p}$.
Note that this holds up to renaming of the items for any pair
$(\RS{A}{i},\REV{\RS{A}{i}})$ or $(\REV{\RS{A}{i}),\RS{A}{i}}$.

We now create another merge. For $p \leq h \leq rp$, we observe that
\[\BLOCK{\sigma_1}{h} = \BLOCK{\sigma_2}{((h-2)\Mod{p})+1} = \cdots = \BLOCK{\sigma_p}{((h-p)\Mod{p})+1}.\]
We can merge each such collection of $p$ identical blocks into
\[ \REP{\REP{e}{p}, \REP{e + 1}{p}, \ldots, \REP{e + (\frac{\ell}{p}-1))}{p}, \REP{e + (\frac{\ell}{p}-1)}{p}, \ldots, \REP{e}{p}}{s}\]
where $e$ is the first item in the blocks.
We can do this $rp-p+1$ times, starting with block $\BLOCK{\sigma_1}{p}$.
This will leave some blocks at the beginning and/or end of each sequence unused
(specifically, the first $p-1$ blocks of $\sigma_1$, the first $p-2$ blocks and last block of $\sigma_2$, the first $p-3$ blocks and last 2 blocks of $\sigma_3$, $\ldots,$ the first block and last $p-2$ blocks of $\sigma_{p-1}$, and the last $p-1$ blocks of $\sigma_p$).

The number of items, $c_1$,
as well as the sum of distances, $c_2$,
in the merge of items in the unused
blocks are a function of $\ell$, $p$, and $s$, but independent of $r$.

For each of the $rp-p$ merges of $p$~blocks,
the first time we consider a given item of
the $\frac{\ell}{p}$ different ones,
the distance to the previous occurrence of that item may be as much as $\ell$.
In the remaining $s-1$ repetitions, the distance varies between
$1$ and $\frac{\ell}{p}$ due to the reversal,
with an average of $\frac{\ell}{2p}+\frac12$,
when we consider a given item again (this average also holds for the second
time one sees an item in the first repetition). Repeating an item $p$ times gives additional
total distance $p-1$ for the $p-1$ repetitions.

In total, the distance is at most
\[ (rp-p)\left(\frac{\ell}{p}\left(\ell+\frac{\ell}{2p}+\frac12\right) + (s-1)\left(\frac{2\ell}{p}\left(\frac{\ell}{2p}+\frac12\right)\right) + s\frac{2\ell}{p}(p-1)\right) + c_2 \]
and the number of items is
\[ (rp-p)s\frac{2\ell}{p}p + c_1 = 2(rp-p)s\ell + c_1, \]
where $c_2$ does not depend on $r$ and $c_1$ does not depend on either $r$ or $s$.

We now consider the limit of the average distances as
$r$ and $s$ approach infinity:
\[ \begin{array}{cl}
  &
   \lim_{s \rightarrow \infty}
   \lim_{r \rightarrow \infty}
   \cfrac{(rp-p)\left(\frac{\ell}{p}\left(\ell+\frac{\ell}{2p}+\frac12\right) + (s-1)\left(\frac{2\ell}{p}\left(\frac{\ell}{2p}+\frac12\right)\right) + s\frac{2\ell}{p}(p-1)\right) + c_2}{
   2(rp-p)s\ell + c_1}
\\[2ex]
= &
   \lim_{s \rightarrow \infty}
   \cfrac{\frac{\ell}{p}\left(\ell+\frac{\ell}{2p}+\frac12\right) + (s-1)\left(\frac{2\ell}{p}\left(\frac{\ell}{2p}+\frac12\right)\right) + s\frac{2\ell}{p}(p-1)}{
     2s\ell}
\\[2ex]
= &
   \cfrac{\frac{2\ell}{p}\left(\frac{\ell}{2p}+\frac12\right) + \frac{2\ell}{p}(p-1)}{2\ell}
\\[2ex]
= &
   \cfrac{\ell+2p^2-p}{2p^2}.
\end{array}
\]
Considering the ratio of the average distances for the two merges, we get
a value of at least
\[ \cfrac{\frac{(2p-1)\ell+p}{2p}}{\frac{\ell+2p^2-p}{2p^2}}
= \cfrac{p(2p-1)\ell+p^2}{\ell+2p^2-p}
= 2p^2 - p - \frac{4(p^4 - p^3)}{\ell + 2p^2 -p}.\]
\end{proof}

\subsection{Upper Bound}

In this section, we prove a matching upper bound, showing that, up to an
additive constant, the maximal ratio for two merges of the same $p$ sequences of
requests to items in a list of length $\ell$ is $2p^2-p$, up to an additive constant,
which depends only on $p$ and $\ell$.
To do this, we first bound how much larger the distance of the concatenation of
the $p$ sequences can be 
compared to
the distance of any merge of these sequences.
\begin{lemma}
\label{Cworst}
  Suppose $C$ is the concatenation of $p$ request sequences
  and $M$ is any merge of the $p$ sequences.
  Then $\TOTAL{C} \leq p \cdot \TOTAL{M}$.
\end{lemma}
\begin{proof}
  We denote the $p$ request sequences by $\sigma_1, \ldots, \sigma_p$.

  For a request sequence $I$ and an item $a$,
  let $\ITEMSUM{I}{a}$ denote the sum of distances of the requests to item $a$,
  i.e., $\ITEMSUM{I}{a} = \sum_{\substack{i\in\SET{1, \ldots, \SIZE{I}} \\ I_i=a}}\REL{I}{i}$.

  Then, $\ITEMSUM{C}{a} \leq \sum_{i=1}^p\ITEMSUM{\sigma_i}{a}$,
  since the only $a$'s changing distance in $C$ are the first ones
  in each of the sequences $\sigma_2,\ldots,\sigma_p$, where their
  distance was the maximal $\ell$.
  Let $\textit{max}$ be the index of a maximal term in this sum, i.e.,
  $\ITEMSUM{\sigma_{\textit{max}}}{a} \geq \ITEMSUM{\sigma_i}{a}$ for any $i$.
  Then $\ITEMSUM{C}{a} \leq p\cdot \ITEMSUM{\sigma_{\textit{max}}}{a}$.

  For any subsequences $S$ and $S'$, and for any $b$, possibly equal to $a$,
  $\ITEMSUM{aSS'a}{a} \leq \ITEMSUM{aSbS'a}{a}$.
  So for any $\sigma_i$, in particular for the maximal index $\textit{max}$ chosen above, we can merge in all items from the other sequences one at a time
  without reducing the sum of the $a$-distances, in the end obtaining $M$.
  Thus, $\ITEMSUM{\sigma_{\textit{max}}}{a} \leq \ITEMSUM{M}{a}$.

  Combining the two inequalities,
  $\ITEMSUM{C}{a} \leq p\cdot \ITEMSUM{M}{a}$, and summing over
  all items gives the result.
\end{proof}

The more difficult result is to bound how much larger the distance of any merge of
the $p$ sequences can be
compared to 
the distance of their concatenation.

\subsubsection{Intuitive Proof Overview}

As intuition for the upper bound in the case where the merge of the $p$
sequences gives a larger total distance than the concatenation, we consider an 
example which is a simplification of the construction used for the lower bound.
Recall that
\[\RS{A}{j} =  (j-1)\frac{\ell}{p}+1, \ldots, j\frac{\ell}{p}, \]
where the concatenation of $\RS{A}{1}, \ldots, \RS{A}{p}$ is $1, \ldots \ell$.
The request sequence of process $j$ 
is $\sigma_j = \RS{B}{j} = \REP{\RS{A}{j} \REV{\RS{A}{j}}}{s}$, which consists of 
$s$ repetitions of the concatenation of $\RS{A}{j}$ and $\REV{\RS{A}{j}}$.
The merge, $M$, we consider is
$ \REP{\RS{A}{1} \cdots \RS{A}{p} \REV{\RS{A}{1}} \cdots \REV{\RS{A}{p}}}{s} $.

In $M$, between consecutive copies of $\RS{A}{j}$ and $\REV{\RS{A}{j}}$ or between $\REV{\RS{A}{j}}$ and $\RS{A}{j}$
from $\sigma_j$, there are requests to $\frac{\ell}{p}$ distinct items from each
of the $p-1$ other sequences. Thus, the distance in $M$ of one request 
coming from $\sigma_j$ to the
previous request to the same item is the sum of the distance between those two
requests in $\sigma_j$ plus $(p-1)\frac{\ell}{p}$.
The asymptotic average distance of the merge, $M$, of the $p$ request
sequences is $\frac{(2p-1)\ell+p}{2p}$.
The asymptotic 
average distance of the concatenation, $C$,
of the $p$ request sequences is $\frac{\ell+p}{2p}$.
Thus, the total distance
in the merge, $\TOTAL{M}$, is, asymptotically, a factor of at most $2p-1$ larger
than the total distance in the concatenation, $\TOTAL{C}$.

We now introduce notation used in the proof.
If a number of sequences including $I$ are merged into $M$,
we let $f^I$ be the function that maps the index of a request
in $I$ to its corresponding index in the merged sequence $M$.
Given two sequences, $I$ and $J$, which, possibly together
with additional sequences, are merged into $M$,
and an index $h$ in $I$,
we define $\NEXTGEN{h}{I}{J}$ to be
\begin{eqnarray*}
\{ j \in \{1,\ldots,|J|\} & | & f^I(h) < f^J(j) < f^I(\SUCC{I}{h}) \mbox{ and }\\
& & J_j \neq J_{j'} \mbox{ for all } j' < j \mbox{ where } f^I(h) <  f^J(j') \},
\end{eqnarray*}
the set of all indices of requests in
$J$ which are merged into $M$ after
the request of index $f^I(h)$, but before the request of index $f^I(\SUCC{I}{h})$,
and which are first occurrences of requests to items with this property.

Recall that
$\REL{\sigma_j}{\SUCC{\sigma_j}{x}}$ is the distance from $\SUCC{\sigma_j}{x}$ to $x$
in $\sigma_j$.
Now we can express the distance from $\SUCC{\sigma_j}{x}$ to $x$, in 
the merged sequence $M$ considered above,
as
$\REL{\sigma_j}{\SUCC{\sigma_j}{x}}+\sum_{k\not= j}\SIZE{\NEXTGEN{x}{\sigma_j}{\sigma_k}}$. 
Summing
this over all $p$ sequences $\sigma_j$ and all $x\in\sigma_j$ gives close to
$(2p-1)\TOTAL{C}$.

More generally, we reduce the problem of proving the upper bound to
proving that $\sum_{x\in I} \SIZE{\NEXTOT{x}}
+ \sum_{x'\in J} \SIZE{\NEXTTO{x'}} \leq 2(\TOTAL{I}+\TOTAL{J})$
for just two sequences, $I$ and $J$.
In our example, this
holds because we can match up sets of requests, since any two indexes $x$ and $x'$ in one 
repetition of $\RS{A}{i}$ have the same sets inserted between them and the next
requests to the same item, so $\SIZE{\NEXTOT{x}} =\SIZE{\NEXTTO{x'}}$.
Thus,
$\sum_{x\in I} \SIZE{\NEXTOT{x}}+ \sum_{x'\in J} \SIZE{\NEXTTO{x'}} \leq
2\sum_{i=1}^s \SIZE{\RS{A}{i}}\cdot\SIZE{\RS{A}{i}}$.

In other cases,
the sizes of the sets of requests merged from one sequence between
requests in another sequence might not all have the same size and they might not even
be merged in as blocks. However, we still partition
the sequences $I$ and $J$ based (in some way) on the sets $\NEXTOT{x}$ and consider the
sizes of the parts of the partitions and multiply them together. If the partitions are
$P^I$ and $P^j$, then we show that $\sum_{x\in I} \SIZE{\NEXTOT{x}}
+ \sum_{x'\in J} \SIZE{\NEXTOT{x'}} \leq 2\sum_{(G,G')\in (P^I,P^J)} \SIZE{G}\cdot\SIZE{G'}$,
where the $G$ and $G'$ are corresponding pairs of parts.

It is easy to see that $2\sum_{(G,G')\in (P^I,P^J)} \SIZE{G}\cdot\SIZE{G'}
\leq \sum_{G\in P^I} \sum\SIZE{G}^2+\sum_{G'\in P^J} \SIZE{G'}^2$. Then, one
notes, using the definitions of the partitions, that the sum of the distances for the 
requests in each set $G$ 
is at least $\frac{1}{2} \sum_{G\in P^I} \SIZE{G}^2$ (and similarly for
$G'$), giving the result. Note that the ordering in our example,
where a set of requests is followed by requests to the same items, but
in the reverse order, gives the minimum distance. 

\subsubsection{The Upper Bound Proof}

First, we show that we can assume that the request sequences are disjoint.
\begin{lemma}
\label{disjointness}
  Given sequences $\sigma_1, \ldots, \sigma_p$ referring to $\ell$ items and any merge $M$, 
there exist
  sequences $\sigma_1', \ldots, \sigma_p'$ and a merge $M'$ such that for each $i$,
  $\TOTAL{\sigma_i}=\TOTAL{\sigma_i'}$, $\TOTAL{M'} \geq \TOTAL{M}$,
  and, for any $j\not=i$, $\sigma_i'$ and $\sigma_j'$ refer to disjoint sets of 
  a total of at most $p\ell$ items.
\end{lemma}
\begin{proof}
Below, we argue that given two lists, we can modify one of them
while retaining its total distance in such a way that the number of items
common to the two sequences is decreased by one and the total distance of the
new, merged sequence $\TOTAL{M'}$ is at least $\TOTAL{M}$.
Using this result inductively, we can make all the sequences disjoint and have
$\TOTAL{M'} \geq \TOTAL{M}$.

Consider two lists $I = I_1,\ldots, I_{n_I}$ and $J = J_1,\ldots, J_{n_J}$
such that there is some item $a$ that occurs in both.

Let $J' = J_1', \ldots, J_{n_J}'$ be the second list in which all occurrences of $a$ 
are
renamed $a'$, where $a'$ does not occur in either list.
Note that $d(J)=d(J')$.

Let $K'=K_1',\ldots,K_{n_I+n_J}'$ be a merge of the lists $I$ and $J'$.

Let $K=K_1,\ldots,K_{n_I+n_J}$ be the list obtained from $K'$ by
replacing all occurrences of $a'$ by $a$. Note that any merge of $I$ and
$J$ can be obtained in this manner.

If $K_j$ and $K_k$ are consecutive occurrences of some item $b \neq a$,
then $K_j'$ and $K_k'$ are also consecutive occurrences of $b$.
Recall that
$\REL{K'}{K_k'} = \SIZE{\SET{K_j',\ldots, K_{k-1}'}}$ and
$\REL{K}{K_k} = \SIZE{\SET{K_j,\ldots, K_{k-1}}}$.
If $\SET{K_j',\ldots, K_{k-1}'}$ does not contain occurrences of both $a$ and
 $a'$,
then $\REL{K'}{K_k'} =\REL{K}{K_k}$; otherwise
$\REL{K'}{K_k'} = 1 + \REL{K}{K_k}$.

Now consider the locations
\[i_{1,1} < \cdots < i_{1,q_1} <i_{2,1} < \cdots < i_{2,q_2} < \cdots 
i_{r,1} < \cdots < i_{r,q_r}\]
of all occurrence of $a$ in $K_1,\ldots,K_{n_I+n_J}$,
where $K_{i_{j,1}}' = \cdots = K_{i_{j,q_j}}'$ for $1 \leq j \leq r$
and $K_{i_{j,1}}' \neq K_{i_{j+1,1}}'$ for $1 \leq j < r$.
Then, for all $1 \leq j \leq r$ and $1 < k \leq q_j$,
$\REL{K'}{K_{i_{j,k}}'} = \REL{K}{K_{i_{j,k}}}$.

For $2< j \leq r$,
\[\begin{array}{rcl}
\REL{K}{K_{i_{j,1}}} & = & \SIZE{\SET{K_{i_{j-1,q_{j-1}}},K_{i_{j-1,q_{j-1}}+1}\ldots,K_{i_{j,1}-1}}} \\[1ex]
& = & 1 + \SIZE{\SETOF{K_h}{i_{j-1,q_{j-1}} < h < i_{j,1}}},
\end{array}\]
whereas
\[\begin{array}{rcl}
\REL{K'}{K_{i_{j,1}}'} & = & 1 + \SIZE{\SETOF{K_h'}{i_{j-2,q_{j-2}} < h < i_{j,1}}} \\[1ex]
 & \geq & 2 + \SIZE{\SETOF{K_h}{i_{j-1,q_{j-1}} < h < i_{j,1}}} \\[1ex]
 & = & 1 + \REL{K}{K_{i_{j,1}}}
\end{array}\]
Note that
$\REL{K}{K_{i_{2,1}}} = 1 + \SIZE{\SETOF{K_h}{i_{1,q_1}  < h < i_{2,1}}}$,
  but $\REL{K'}{K_{i_{2,1}}'}=\ell$, since it is the first
  occurrence of this item.

  Thus,
\[
  \TOTAL{K'} \geq \TOTAL{K}+\ell - \REL{K}{K_{i_{2,1}}}  \geq \TOTAL{K}.
\]

Introducing the new item $a'$ in the above increases the total number
of items by one. Doing so $\ell$ times for $p-1$ sequences
increases the number of items to at most $p\ell$.
\end{proof}

For the proof sequence to follow,
we first consider two request sequences and then later
reduce the proof for $p$ sequences to the result for two sequences.

We use the following algorithms,
Algorithms~\ref{partition-creation-one} and~\ref{partition-creation-two},
to define partitions $P^I$ and $P^J$ of the indices of
two request sequences, $I$ and $J$,
such that all indices in a partition represent requests to distinct items,
i.e., if indices $i$ and $j$ are in the same part of a partition of $I$,
then $I_i\not=I_j$.
We slightly abuse the term partition since not all requests
in $J$ are necessarily included in a part in $P^J$, but $P^J$ is
a partition of a subsequence of $J$.
In the algorithms,
we refer to an index as \emph{unassigned} as long as it has not
been assigned to any part in a partition.

\newcommand\CONDITION[2]%
  {\begin{tabular}[t]{@{}l@{}l@{}}
     #1&#2
   \end{tabular}%
  }

\begin{algorithm}[htbp]
  \caption{Creating the partition for $I$.}
  \label{partition-creation-one}
\begin{algorithmic}[1]
\Let{$P^I$}{$\emptyset$}
\Let{$\INDEX$}{$0$}
\While{there are unassigned indices in $I$}
  \Let{$i$}{first such index}
  \Let{$\INDEX$}{$\INDEX + 1$}
  \Let{$\PART{I}_{\INDEX}$}{$\SET{i}$}
  \If{there exists an index $j>i$ such that $I_i=I_j$}
    \Let{$j$}{smallest such index}
    \For{$h$ in $i+1, \ldots, j-1$}
      \If{$h$ is unassigned \textbf{and} $I_h\not\in\SET{I_{i+1}, \ldots, I_{h-1}}$}
        \Let{$\PART{I}_{\INDEX}$}{$\PART{I}_{\INDEX} \cup \SET{h}$}
      \EndIf
    \EndFor
  \EndIf
  \Let{$P^I$}{$P^I \cup \SET{\PART{I}_{\INDEX}}$}
\EndWhile
\end{algorithmic}
\end{algorithm}

\begin{algorithm}[htbp]
  \caption{Creating the partition for $J$ based on $I$.}
  \label{partition-creation-two}
\begin{algorithmic}[1]
\Let{$P^J$}{$\emptyset$}
\Let{$q$}{number of parts in $P^I$}
  \For{$\INDEX$ in $1, \ldots, q$}
  \Let{$\PART{J}_{\INDEX}$}{$\emptyset$}
  \Let{$i_1, \ldots, i_q$}{ordered sequence of indices in $\PART{I}_{\INDEX}$}
  \For{$j$ in $1, \ldots, q$}
  \Let{$\PART{J}_{\INDEX}$}{
    $\PART{J}_{\INDEX} \cup 
        \SETOF{h \in \NEXTOT{i_j}}{\mbox{$h$ is unassigned } \textbf{and }
        J_h \not\in \bigcup_{k\in\SET{1, \ldots, j-1}}\SETOF{J_l}{l \in \NEXTOT{i_k}} }$}
    \label{part-addition}
  \EndFor
  \Let{$P^J$}{$P^J \cup \SET{\PART{J}_{\INDEX}}$}
\EndFor
\end{algorithmic}
\end{algorithm}

Thus, with reference to two fixed request sequences $I$ and $J$,
a merge of them, $M$, and the partitions defined by the algorithms,
the parts of $P^I$ are enumerated in order of creation as
$\PART{I}_1, \PART{I}_2, \ldots, \PART{I}_q$.
The creation of each part in $P^J$ is triggered by a part of $P^I$.
For these parts, $\PART{J}_1, \PART{J}_2, \ldots, \PART{J}_q$,
the $\PART{I}_j$ triggered the creation of $\PART{J}_j$,
and we refer to these as \emph{corresponding} parts.

We let \PP denote the set of all pairs of indices in the Cartesian
product of corresponding parts. Thus, if there are $q$ parts in $P^I$,
\[\PP=\SETOF{(i,j)}{h \in \SET{1, \ldots, q}, i \in \PART{I}_h, j \in \PART{J}_h}.\]

\begin{lemma}
  \label{cardinality-equality}
  Defining sets and partitions,
  following Algorithms~\ref{partition-creation-one} and~\ref{partition-creation-two},
  from two disjoint request sequences $I$ and $J$, which are merged into $M$,
  there exists an injective mapping from \\
  $\SETOF{(x,y)}{x\in \SET{1, \ldots, \SIZE{I}}, y\in\NEXTOT{x}}$
  to \PP.
\end{lemma}
\begin{proof}
  Assume that $y\in\NEXTOT{x}$. Since $P^I$ is a partition,
  there exists a unique part $\PART{I}\in P^I$ such that $x\in \PART{I}$.
  Let $\PART{J}\in P^J$ be the corresponding part, i.e., the part,
  the creation of which was
  triggered by $\PART{I}$ in Algorithm~\ref{partition-creation-two}.

  \emph{Case $y\in \PART{J}$}:\\
  If $y\in \PART{J}$, we map $(x,y)$ to $(x,y)$.
  For the remainder of the proof, we assume that $y\not\in \PART{J}$.

  \emph{Case $y\not\in \PART{J}$ and $y$ is not in any part of $P^J$}:\\
  Since $y\in\NEXTOT{x}$, $y$ is the index from $J$ to the first occurrence of
  an item $J_y$ appearing after $f^I(x)$ in $M$.
  We first observe that the index $x$ cannot be the first
  added to its part in Algorithm~\ref{partition-creation-one}.
  If it were, then $y$ would be added to a part in
  Algorithm~\ref{partition-creation-two}, since $y$ is unassigned
  (by assumption) and the rightmost argument is trivially true
  the first time through the \textbf{for}-loop
  (Line~\ref{part-addition} of Algorithm~\ref{partition-creation-two}),
  so $y$ would be in $\PART{J}$.
  Thus, assume an earlier index $z$ into $I$ was the first
  to be placed in the part where $x$ was added later.

  Since $x$ is added to $z$'s part, there are no other requests
  to the item $I_x$ between $z$ and $x$ in $I$.

  Since $y$ is not in any part, that is, unassigned, 
  there must be another request to $J_y$ between $f^I(z)$ and $f^I(x)$ in $M$,
  preventing $y$ from being added.
  Let $y'$ be the first index into $J$ such that $f^J(y')$ comes after
  $f^I(z)$ in $M$ and $J_y=J_{y'}$.
  By definition, since $I$ and $J$ are disjoint, $y'\in \PART{J}$.
  We map $(x,y)$ to $(x,y')$, and since $y'\not\in\NEXTOT{x}$, no other
  pair can be assigned to $(x,y') \in \PP$.

  \emph{Case $y\not\in \PART{J}$, but $y$ is in some part}:\\
  Let $z$ be the request in $I$ which starts the part $\PRIMEPART{I}$ in $P^I$
  which gives rise to the part $\PRIMEPART{J}$ in $P^J$ to which $y$ belongs.
  There must be a request to item $I_x$ between $z$ and $x$,
  since otherwise, $x$ would be in the same part as $z$ and
  $y$ would then belong to $\PART{J}$.
  Let $x'$ be the first index after $z$ such that $I_x=I_{x'}$.
  We assign $(x,y)$ to $(x',y)$.
  Since $y$ belongs to the part of $P^J$, the creation of which was
  triggered by the part containing $z$, there cannot be any request
  to any item equal to $J_y$ between $f^I(z)$ and $f^J(y)$ in $M$, so no other
  pair can be assigned to $(x',y)$.
  Since $x' \in \PRIMEPART{I}$, $x'$ and $y$ are again in corresponding parts.
\end{proof}

\begin{lemma}
  \label{cardinality-bound}
  For two disjoint request sequences $I$ and $J$ together with their merge $M$,
  \[\SIZE{\SETOF{(x,y)}{y\in\NEXTOT{x}}} \geq \SIZE{\SETOF{(x,y)}{x\in\NEXTTO{y}}} - \ell^2.\]
\end{lemma}
\begin{proof}
  Consider a pair $(x,y)$ such that $x\in\NEXTTO{y}$.
  Assume that there are requests to the item $I_y$ after $y$ in $J$.
  Let $y'$ be the first index larger than $y$ in $J$ such that
  $I_{y'}=I_y$ and let $x'$ be the last index in
  $I$ such that $I_{x'}=I_x$ and $f^I(x')$ comes before $f^J(y')$ in $M$
  ($x'$ could be $x$).
  Clearly, $y'\in\NEXTOT{x'}$, so $(x',y')\in\SETOF{(x,y)}{y\in\NEXTOT{x}}$.
  The mapping $(x,y) \mapsto (x',y')$ is injective since no index
  $x''$ between $x$ and $x'$ such that $I_{x''}=I_x$ can belong
  to $\NEXTTO{y}$.
  The mapping is defined for all but at most $\ell$ requests in $J$ (the last
  request to each item) and at most $\ell$ $x$-values for each such request
  $y$ in $J$, which amounts to at most $\ell^2$ pairs.
\end{proof}

\begin{lemma}
  \label{set-partition-bound}
  For two disjoint request sequences $I$ and $J$ together with their merge $M$
  and partitions $P^I$ and $P^J$,
  \[\sum_{1\leq x\leq\SIZE{I}}\SIZE{\NEXTOT{x}} + \sum_{1\leq y\leq \SIZE{J}}\SIZE{\NEXTTO{y}} \leq
  2 \SIZE{\PP} + \ell^2.\]
\end{lemma}

\begin{proof}
  $\sum_{1\leq x\leq\SIZE{I}}\SIZE{\NEXTOT{x}} =
  \SIZE{\SETOF{(x,y)}{x\in \SET{1, \ldots,\SIZE{I}}, y\in\NEXTOT{x}}}$,
  which, by Lemma~\ref{cardinality-equality}, is at most $\SIZE{\PP}$.
  
  By Lemma~\ref{cardinality-bound},
  $\SIZE{\SETOF{(x,y)}{y\in\NEXTOT{x}}} \geq \SIZE{\SETOF{(x,y)}{x\in\NEXTTO{y}}} - \ell^2$.
  So, we get that $\sum_{1\leq y\leq\SIZE{J}}\SIZE{\NEXTTO{y}} \leq \SIZE{\PP} + \ell^2$.

  Adding the two bounds gives the result.
\end{proof}

We need the following simple lemma:
\begin{lemma}
\label{lemma-reverse}
If $I$ is a sequence of requests to distinct items and $J$ can be
any permutation of $I$, then $\REL{JI}{I}$ is minimized when $J=\REV{I}$.
\end{lemma}
\begin{proof}
  Suppose to the contrary that the minimum is attained for some
  $J\not=\REV{I}$.

  Let $J=(x_1,x_2,\ldots x_m)$.
  Consider the largest value $k$ such that $x_k$ and $x_{k+1}$ are
  indices of items occurring in the same order in $I$. Define
  \[J'=(x_1,\ldots, x_{k-1},x_{k+1},x_k,x_{k+2},...,x_m),\]
   where only these two items are swapped. Then,
  \[\begin{array}{rcl}
     \REL{J'I}{I} & = & \REL{JI}{I}
     + (\REL{J'I}{x_{k+1}} - \REL{JI}{x_{k+1}})
     + (\REL{J'I}{x_{k}} - \REL{JI}{x_{k}}) \\[1ex]
  & = & \REL{JI}{I} - 1
  \end{array}\]

  This follows since items originating from the $x_k$'s no longer have
  the $x_{k+1}$-items between them, and the
  two copies of $x_k$-items between the $x_{k+1}$-items are only counted once.
  Thus, $J$ did not give rise to the minimum value, giving the
  contradiction.
\end{proof}

\begin{lemma}
  \label{partition-distance-bound}
  For two request sequences $I$ and $J$ together with their merge $M$
  and partitions $P^I$ and $P^J$,
  $\sum_i \SIZE{\PART{I}_i} \cdot \SIZE{\PART{J}_i} \leq
    \sum_i \REL{I}{\PART{I}_i} +
    \sum_i \REL{J}{\PART{J}_i}+3\ell^2. $
\end{lemma}
\begin{proof}
  By simple arithmetic, for any two positive values,
  $a$ and $b=a-c$ for some $c$,
  $2ab = 2a^2-2ac \leq 2a^2 -2ac+b^2 = a^2+b^2$.
  By this, we conclude that
  \begin{eqnarray}
    \label{from-arithmetic}
    2\sum_i \SIZE{\PART{I}_i} \cdot \SIZE{\PART{J}_i} \leq
    \sum_i \SIZE{\PART{I}_i}^2 + \sum_i \SIZE{\PART{J}_i}^2.
  \end{eqnarray}
  We can relate $\SIZE{\PART{I}_i}^2$ to $\REL{I}{\PART{I}_i}$,
  and similarly for $J$.
  We prove the following properties where $K$ can be either $I$ or $J$:
  \begin{itemize}
  \item The items requested in $\PART{K}_i$ are all distinct.
  \item If $q$ is the first index added to $\PART{K}_i$, then the
    other items requested in $\PART{K}_i$ are each the first occurrence
    of that item after $q$ in $K$.
  \end{itemize}
  For $\PART{I}_i$, the index $h$ is only added if
  $I_h \not\in \SET{I_{q+1} ,\ldots, I_{h-1} }$,
  where $h$ is the next index for a request to $I_q$, and the indices
  $q+1 ,\ldots, h-1$ are the only smaller indices considered
  for $\PART{I}_i$. Thus, both properties hold for $K=I$.
  
  For $\PART{J}_i$, the index $h$ is only added if
  $$J_h \not\in \bigcup_{k\in\SET{1, \ldots, j-1}}\SETOF{J_l}{l \in \NEXTOT{i_k}},$$
  where the items indexed in each $\NEXTOT{i_k}$ are distinct and
  $\bigcup_{k\in\SET{1, \ldots, j-1}}\NEXTOT{i_k}$
  are the only smaller indices considered for $\PART{J}_i$, other than
  smaller indices in $\NEXTOT{i_j}$. Thus, the first property holds.

  For the second property,
  assume for the sake of contradiction that $q$ is the smallest index added to $\PART{J}_i$, 
  $q<j_1<j_2$, $J_{j_1}=J_{j_2}$, with $j_1$ being the smallest index having this
  property, and $j_2 \in \PART{J}_i$.
  Let $x$ be the smallest index in $\PART{I}_i$.
  Now, $j_1$ is not in $\NEXTOT{x}$, since then $j_2$ could not belong to $\PART{J}_i$,
  by Algorithm~\ref{partition-creation-two}.
  Hence, there exists an index $x'$ with $I_x=I_{x'}$ after $x$, which is merged
  into $M$ before $j_1$.
  For $j_2$ to belong to $\PART{J}_i$, there must exist a $y$ in $I$ between
  $x$ and $x'$ such that $j_2$ is in $\NEXTOT{y}$.
  However, since $y$ appears before $x'$, $j_1$ would be in $\NEXTOT{y}$,
  preventing $j_2$ from belonging to $\NEXTOT{y}$, arriving at a contradiction.

  The following argument holds for any of the parts from either $I$ or $J$
  and we use $K$ to denote either.

  Consider the sequence $X^K_i$ of requests in $\PART{K}_i$,
  for which there is a previous request to the same item in $K$,
  and let the requests in $X^K_i$ occur in the order they occur in $K$.
  Further, let $x$ be the smallest index among them. The requests in $\PART{K}_i$,
  and therefore also in $X^K_i$, are all to distinct items,
  and they are each the first request to that item after
  index $x$. This means that for any request of $X^K_i$, the previous
  request to it, if any, occurred before the request at index $x$.
  
  The total distance between the requests indexed by $X^K_i$ 
  and their previous
  occurrences is at least $a=\min_Y \REL{YX^K_i}{X^K_i}$,
  where $Y$ is any permutation of $X^K_i$.
  By Lemma~\ref{lemma-reverse}, the minimum value, $a$, is obtained when $Y$ is
  $\REV{X^K_i}$.

For the special permutation $\REV{X}$,
$\REL{\REV{X}X}{X} = \SIZE{X} \cdot (\SIZE{X}+1)/2$.
Thus, for any $X^K_i$, one can assume that
$\SIZE{X^K_i}^2 \leq 2 \REL{K}{X^K_i}$.

Now,
\[\begin{array}{rcl}
&& \sum_i \SIZE{\PART{I}_i}^2 + \sum_i \SIZE{\PART{J}_i}^2 \\[1ex]
 & = &
\sum_i \SIZE{X^I_i}^2 + \sum_i \SIZE{\PART{I}_i\setminus X^I_i}^2 + 2\sum_i \SIZE{\PART{I}_i}\cdot\SIZE{\PART{I}_i\setminus X^I_i} + \mbox{} \\[1ex]
&& 
\sum_i \SIZE{X^J_i}^2 + \sum_i \SIZE{\PART{J}_i\setminus X^J_i}^2 + 2\sum_i \SIZE{\PART{J}_i}\cdot\SIZE{\PART{J}_i\setminus X^J_i} \\[2ex]
 & \leq &
\sum_i \SIZE{X^I_i}^2 +3\sum_i \SIZE{\PART{I}_i}\cdot\SIZE{\PART{I}_i\setminus X^I_i} + \mbox{} \\[1ex]
&& 
\sum_i \SIZE{X^J_i}^2 +3\sum_i \SIZE{\PART{J}_i}\cdot\SIZE{\PART{J}_i\setminus X^J_i} \\[2ex]
 & \leq &
\sum_i \SIZE{X^I_i}^2 +3\ell^2 +
\sum_i \SIZE{X^J_i}^2 +3\ell^2\\[2ex]
& \leq &
\sum_i 2\REL{I}{X^I_i} + \sum_i 2\REL{J}{X^J_i}
+6\ell^2 \\[2ex]
& \leq &
2\sum_i \REL{I}{\PART{I}_i} +
2\sum_i \REL{J}{\PART{J}_i}+6\ell^2.
\end{array}\]
In the first equality, we abuse notation slightly and treat
sequences as sets in the obvious way.
The second inequality follows
since, for each part, there are at most $\ell$ requests in total from
all $\PART{K}_i$ together which are not counted in any $X^K_i$
because they did not have previous requests, and since each
$\PART{K}_i$ has size at most $\ell$.

Combining this with Eq.~\ref{from-arithmetic},
\[2\sum_i \SIZE{\PART{I}_i} \cdot \SIZE{\PART{J}_i} \leq
2\sum_i \REL{I}{\PART{I}_i} + 2\sum_i \REL{J}{\PART{J}_i}+6\ell^2,\]
and the lemma follows.
\end{proof}

\begin{lemma}
\label{Cbest}
  For any integer $p\geq 2$,
  there exists a constant $c$, depending on $p$ and $\ell$, such that for all
  disjoint request sequences $\sigma_1, \sigma_2, \ldots, \sigma_p$,
  \[\TOTAL{M} \leq (2p-1)\TOTAL{C} + c,\]
  where $C$ is the concatenation of $\sigma_1, \sigma_2, \ldots, \sigma_p$ and
  $M$ is any merge of these $p$ sequences.
\end{lemma}
\begin{proof}
For the concatenation, $C$, of the $p$ sequences, the
sum of all the distances in $C$ is
$\TOTAL{C} \leq \sum_{i=1}^p \TOTAL{\sigma_i}$,
because the sequences are disjoint.

Recall that the distance of some index~$h$ in $\sigma_j$,
$\REL{\sigma_j}{h}$, is the number of distinct requests
between $h$ and $\PREV{\sigma_j}{h}$ if $\PREV{\sigma_j}{h}$ exists
and $\ell$ otherwise.

Due to disjointness, this distance in $\sigma_j$
is only increased in $M$ by the requests from other
sequences that are inserted between those two requests.
Thus, the distance between $f^{\sigma_j}(h)$ and $f^{\sigma_j}(\PREV{\sigma_j}{h})$,
is at most
$\REL{\sigma_j}{h}+\sum_{k\not= j}\SIZE{\NEXTGEN{\PREV{\sigma_j}{h}}{\sigma_j}{\sigma_k}}$.
If $\PREV{\sigma_j}{h}$ does not exists, we define
$\SIZE{\NEXTGEN{\PREV{\sigma_j}{h}}{\sigma_j}{\sigma_k}}$ to be zero.

Using this,
\[\begin{array}{rcl}
\TOTAL{M}
     & =    & \sum_{i\in\SET{1,\ldots,\SIZE{M}}} \REL{M}{i} \\[1ex]
     & =    & \sum_{j=1}^p \sum_{h\in\SET{1,\ldots,\SIZE{\sigma_j}}} \REL{M}{f^{\sigma_j}(h)} \\[1ex]
     & \leq &
        \sum_{j=1}^p \sum_{h\in\SET{1,\ldots,\SIZE{\sigma_j}}} (\REL{\sigma_j}{h}+\sum_{\substack{k=1 \\ k\not=j}}^p\SIZE{\NEXTGEN{\PREV{\sigma_j}{h}}{\sigma_j}{\sigma_k}})\\[1ex]
     & =    &
         \sum_{j=1}^p (\TOTAL{\sigma_j}+\sum_{h\in\SET{1,\ldots,\SIZE{\sigma_j}}} \sum_{\substack{k=1 \\ k\not=j}}^p\SIZE{\NEXTGEN{\PREV{\sigma_j}{h}}{\sigma_j}{\sigma_k}})
\end{array}\]

By the interpretation of
$\SIZE{\NEXTGEN{\PREV{\sigma_j}{h}}{\sigma_j}{\sigma_k}}$
when $\PREV{\sigma_j}{h}$ does not exist,
\[
\sum_{h\in\SET{1,\ldots,\SIZE{\sigma_i}}} \sum_{\substack{k=1 \\ k\not=i}}^p\SIZE{\NEXTGEN{\PREV{\sigma_j}{h}}{\sigma_i}{\sigma_k}}
\leq
\sum_{h\in\SET{1,\ldots,\SIZE{\sigma_i}}} \sum_{\substack{k=1 \\ k\not=i}}^p\SIZE{\NEXTGEN{h}{\sigma_i}{\sigma_k}}.
\]
Since $\sum_{i=1}^p \TOTAL{\sigma_i} \leq \TOTAL{C}$, as stated as the
first remark of the proof,
we are done if we can establish that
\[\sum_{i=1}^p (\TOTAL{\sigma_i}+\sum_{h\in\SET{1,\ldots,\SIZE{\sigma_i}}} \sum_{\substack{k=1 \\ k\not=i}}^p\SIZE{\NEXTGEN{h}{\sigma_i}{\sigma_k}}) \leq (2p-1)\sum_{i=1}^p \TOTAL{\sigma_i}+7p^2\ell^2,\]
where $7p^2\ell^2$ is the constant $c$ in the lemma statement.
This is equivalent to showing that
\[\sum_{h\in\SET{1,\ldots,\SIZE{\sigma_i}}} \sum_{\substack{k=1 \\ k\not=i}}^p\SIZE{\NEXTGEN{h}{\sigma_i}{\sigma_k}} \leq
 2(p-1)\sum_{i=1}^p \TOTAL{\sigma_i}+7p^2\ell^2.\]

For each sequence $\sigma_j$, we consider all $p-1$ sequences $\sigma_k$ for
which the cardinality of the sets
$\NEXTGEN{h}{\sigma_j}{\sigma_k}$ and $\NEXTGEN{h'}{\sigma_k}{\sigma_j}$ are added into the sum
for some $h$ and $h'$.

From the definition of $\PP$, we have that
$\SIZE{\PP} = \sum_i \SIZE{\PART{I}_i} \cdot \SIZE{\PART{J}_i}$.
Using that in combination with
Lemmas~\ref{set-partition-bound} and~\ref{partition-distance-bound}
gives that for all pairs, $j$ and $k$,
\[
\sum_{h\in\SET{1,\ldots,\SIZE{\sigma_j}}}\SIZE{\NEXTGEN{h}{\sigma_j}{\sigma_k}}+\sum_{h'\in\SET{1,\ldots,\SIZE{\sigma_k}}}\SIZE{\NEXTGEN{h'}{\sigma_k}{\sigma_j}}
 \leq 2(\TOTAL{\sigma_j}+\TOTAL{\sigma_k}) + 7\ell^2.
\]

If we sum this over all such pairs $j$ and $k$, each $j$ will appear with
$p-1$ different indices $k$, giving
\[\sum_{h\in\SET{1,\ldots,\SIZE{\sigma_j}}}\sum_{\substack{k=1 \\ k\not=j}}^p\SIZE{\NEXTGEN{h}{\sigma_k}{\sigma_j}}
\leq 2(p-1)\sum_{i=1}^p \TOTAL{\sigma_i} + 7p^2\ell^2\]
\end{proof}

\begin{theorem}
\label{theorem-ratio-upper-bound}
  For any two merges, $M_1$ and $M_2$ of $p$ sequences of requests
  to a list of length~$\ell$, the total distance of $M_1$ is at most
  $2p^2-p$ times the total distance of $M_2$ up to an additive constant
  only depending on $p$ and $\ell$.
\end{theorem}
\begin{proof}
  By Lemma~\ref{disjointness}, we can assume disjointness at 
  a loss of an additive constant only depending on $p$ and $\ell$.

  Let $C$ be the concatenation of the disjoint sequences $\sigma_1, \ldots, \sigma_p$.

  By Lemma~\ref{Cworst}, $\TOTAL{C} \leq p \cdot \TOTAL{M_2}$.

  By Lemma~\ref{Cbest}, $\TOTAL{M_1} \leq (2p-1) \TOTAL{C} + c$,
  where $c$ is a constant only depending on $p$ and $\ell$.

  Combining these two inequalities gives the result.
\end{proof}

\section{Interleaving at the Item Access Level -- Fully Adversarial}
\label{sec:itemLevel}
In the previous section, we proved a bound on \DMTF's cost
relative to \OPT when operations are interleaved at
the \emph{operation level},
i.e., after the merging of the $p$ request sequences, only
the cost of the sequential \MTF was counted.

Now, we consider a more fine-grained interleaving; more than
one process might be searching for the same item at the same
time. 
Potentially, moving from the operation level to the item level
could lead to wasted work. For instance, if two
processes are searching for the same item $a$, currently in position
$k$, then the cost (the number of items inspected) would be $k+1$
when interleaving at the operation level, since one process would
carry out its operation first, moving $a$ to the front
before the other process would carry out its operation.
The more fine-grained interleaving allows both processes to
inspect the first item, then the second item, etc.
Thus, the total number of items inspected could be at least~$2k$.

We assume that these processes will find the item in
the same location and have the same cost, rather than all but
the first having cost $1$, as in the operation level.
We refer to this level as the \emph{item access level}.
We emphasize that we are not yet considering concurrency
where interleaving will be even more fine-grained.

In this section,
we prove that the bound derived in the previous section also
holds when considering the more fine-grained interleaving.
Intuitively, this is because extra work of the type described above
in the fine-grained interleaving occur in situations
that are not worst-case. 

\begin{theorem}
\label{theorem-item-level}
  Considering item level interleaving,
  for any two merges, $M_1$ and $M_2$ of $p$ sequences of requests
  
   and any scheduling of steps to
  \DMTF and \OPT,
  \[\DMTF(M_1) \leq (4p^2-2p)\OPT(M_2) + O(1),\]
  where
  the constant depends only on $p$ and $\ell$.
\end{theorem}
\begin{proof}
Assume that some process $p_i$ treats a request to $a$, finds it,
and moves it to the front.
Assume that some other process $p_j$ initiates a search
for $a$ before $a$ is moved to the front and stops its search
after $a$ has been moved to the front.
We consider a linearization
(as defined in Section~\ref{linearization_section})
based on the order of
moves to the front or discovering that an item has already been
moved to the front.
Assume that in this linearization, the $a$ of process $p_j$
is preceded by a request $b\not=a$.
The linearization where we swap $a$ and $b$ can only have higher cost,
since $p_j$'s search for $a$ could now pass through item~$b$,
since it is not moved to the front until after $p_j$ stops its search.
Using this argument inductively, we can consider a worst-case
linearization where where all these requests to $a$ appear
in succession. We let $M$ refer to this sequence.

Our upper bound from Theorem~\ref{theorem-ratio-upper-bound}
in the previous section also holds when items
are renamed so that the items requested in each process are
disjoint.
Consider the sequence $M'$, where $a$'s from process $p_i$,
$i\in\SET{1,\ldots, p}$ are renamed $a_i$, where $a_1, \ldots, a_p$
is a fresh set of names.

Since $p_i$ moved $a$ to the front, $a_i$ appears first among
the requests to $a_1, \ldots, a_p$ under discussion, say at index $y$
in $M'$. Let $a_h$ be the request prior to $a_i$ which
appears latest in $M'$ at index $x<y$.
Thus, $\REL{M'}{y} \geq \SIZE{\SET{M_x, \ldots, M_y}}$
and the same lower bound holds for the other requests to
$a_1, \ldots, a_{i-1}, a_{i+1}, \ldots, a_p$.

In $M$, $\SIZE{\SET{M_x, \ldots, M_y}}$ is the exact cost of
processing the request to $a$ at index $x$.
For item access level interleaving, for the other requests to~$a$, we
also count cost  $\SIZE{\SET{M_x, \ldots, M_y}}$.
Thus, 
the extra cost is already accounted for in the distances
used in the proof.

By Theorem~\ref{theorem-reduced-to-ratio}, the ratio at the item access level is bounded
by at most twice the ratio $2p^2-p$ from Theorem~\ref{theorem-ratio-upper-bound}.
\end{proof}

\section{ \DMTF: A Distributed Implementation of Move-to-Front}
\label{implementation}

This section presents \DMTF, a lock-free implementation of a static set
based on the well-known \MTF algorithm.
At the end of the description of the algorithm, a simple
modification to make it wait-free is described.
A complete proof of correctness is presented in in Section~\ref{correctness}.

When a process finds a node containing the item it is looking for
and that node is not at the front of the list,
it prepends a copy of the node to the list
and then removes the node it found from the list.
This ensures that the list always contains a node containing each item in the set.

There are difficulties with implementing the \MTF algorithm in this straightforward fashion in a distributed setting.
For example, suppose that some process $p_i$ has begun looking for an item and
has proceeded a few nodes along the list.
If another process, which is concurrently looking for the same item, finds a node containing that item later in the list and moves it to the front of the list, then $p_i$ might reach the end of the list without finding the item.
To prevent this problem,
each process announces the item it is currently looking for.
After a process finds the item, $e$,
it is looking for and prepends a new node containing $e$ to the list,
it informs all other processes that are looking for $e$ about this new node
before removing the node in which it found $e$ from the list.

The processes also help one another to prepend nodes to the list and remove nodes from the list.
This ensures that more than one copy of a node is not prepended to the list,
for example, when multiple processes concurrently find the same node
or when a process has fallen asleep for a long time and then
wakes up. Additional fields in each node
are used to facilitate helping.

A detailed description of our implementation
is given below.
Pseudocode is presented in Figures \ref{searchcode} and \ref{mtfcode}.
Throughout the code, 
if $h$ is a pointer to a node and $f$ is the name of the field, then
$h.f$ is a reference to that field of the node.
In the description and proof of correctness, we distinguish between nodes and pointers
to nodes.

\begin{figure}[htp]
{\footnotesize
\begin{code}
{\bf function} SEARCH($e$) by process $p_i$\ul
\n $\rhd$ Allocate and initialize a new node.\nl
$g \leftarrow$ allocate(node)\label{allocate}\nl
$g.item \leftarrow e$\label{newitem}\nl
$g.prev \leftarrow g.next \leftarrow g.old \leftarrow g.new \leftarrow \NULL$\label{newptrs}\ul
$\rhd$ Announce the operation.\nl
CAS($A[i], (\NULL,\bot), (g,e)$)\label{announce}\nl
$(h_1,h) \leftarrow Head$\label{gethead}\nl
{\bf if} $h_1.item = e$ {\bf then} \label{katfront}\ul
\n $\rhd$ The first node in the list contains item $e$. Moving it to the front is not necessary.\nl 
$(a,b) \leftarrow$ CAS($A[i], (g,e), (\NULL,\bot)$)\label{unannFront1}\nl
{\bf if} $b = \bot$ {\bf then}  CAS($A[i],(a,b),(\NULL,\bot)$)\label{unannFront2}\nl
{\bf return}($h_1$)\label{return1}\ul
\p $\rhd$ Continue searching from the second node in the list.\nl
$c \leftarrow 0$\nl
{\bf while}  $h \neq \NULL$ {\bf do}\label{htest}\nl
\n {\bf if} $h.item = e$ {\bf then} \label{kfoundlater}\ul
\n $\rhd$ A node containing item $e$ is found, but not at the front of the list.\nl
$(a,b) \leftarrow A[i]$\label{readann1}\nl
{\bf if} $b = \bot$ {\bf then} \label{checkOther1}\ul
\n $\rhd$ Some other process has changed the announcement.\nl
CAS($A[i],(a,b),(\NULL,\bot)$)\label{unannOther1}\nl
{\bf return}($a$)\label{return2}\ul
\p $\rhd$ Try to set the nodes pointed to by $g$ and $h$ to refer to one another.\nl
CAS($g.old, \NULL, h$)\label{setgold}\nl
CAS($h.new, \NULL, g$)\label{sethnew}\ul
$\rhd$ Move the new copy of the node containing item $e$ to the front of the list.\nl
$g' \leftarrow h.new$\label{setgprime}\nl
{\bf if} $g' \neq \GONE$ {\bf then}  
MOVE-TO-FRONT($g'$)\label{mtf}\nl
$(a,b) \leftarrow A[i]$\label{readann2}\nl
CAS($A[i],(a,b),(\NULL,\bot)$)\label{unannPostMTF}\nl
{\bf return}($a$)\label{return3}\nl
\p $c \leftarrow  (c + 1) \bmod \phi$\nl
{\bf if} $c = 0$ {\bf then} \nl
\n $(a,b) \leftarrow A[i]$\label{readann3}\nl
{\bf if} $b = \bot$ {\bf then} \label{checkOther2}\ul
\n $\rhd$ Some other process has changed the announcement.\nl
CAS($A[i],(a,b),(\NULL,\bot)$)\label{unannOther2}\nl
{\bf return}($a$)\label{return4}\nl
\p\p $h \leftarrow h.next$\label{next}\nl
\p {\bf end while} \ul
$\rhd$ End of the list was reached.\nl
$(a,b) \leftarrow A[i]$\label{readann4}\nl
CAS($A[i],(a,b),(\NULL,\bot)$)\label{unannEnd}\ul
$\rhd$ If no other process has changed the announcement, item $e$ is not in the list.\nl
{\bf if} $b = \bot$ {\bf then}  {\bf return}($a$) {\bf else} {\bf return}({\sc not present})\label{return5}
\end{code}
}
\caption{An Algorithm to Search for a Node Containing the Element $e \neq \bot$}
\label{searchcode}
\end{figure}

\begin{figure}[htp]
{\footnotesize
\begin{code}
{\bf procedure} MOVE-TO-FRONT($g'$) by process $p_i$\nl
{\bf while}  $g'.old \neq \DONE$ {\bf do}\label{outerloop}\nl
\n $(h_1,h_2) \leftarrow Head$\label{gethead2}\nl
$h' \leftarrow h_1.old$\label{getold}\nl
{\bf if} $h' = \DONE$ {\bf then} \label{testold}\ul
\n $\rhd$ Try to prepend the node pointed to by $g'$ to the beginning of the list\nl
{\bf if} $g'.old \neq \DONE$ {\bf then}  CAS($Head,(h_1,h_2),(g',h_1)$)\label{tryprepend}\nl
\p {\bf else} $\rhd$ Ensure that the first two nodes in the list point to one another.\nl
\n CAS($h_1.next, \NULL, h_2)\label{setnext}$\nl
CAS($h_2.prev, \NULL, h_1)$\label{setprev}\nl
$e' \leftarrow h_1.item$\label{h1item}\ul
$\rhd$ Inform all processes looking for item $e'$\nl
{\bf for} every process index $j$ {\bf do}\nl
\n $(a,b) \leftarrow A[j]$\label{readannMTF}\nl
{\bf if} $b = e'$ and $h_1.old \neq \DONE$ {\bf then}  CAS($A[j], (a,b), (h_1,\bot)$)\label{inform}\ul
\p $\rhd$ Remove the node pointed to by $h'$ from the list.\nl
$pred \leftarrow h'.prev$\label{getprev}\nl
$succ \leftarrow h'.next$\label{getnext}\nl
CAS($pred.next, h',succ$)\label{remove1}\nl
{\bf if} $succ \neq \NULL$ {\bf then}  CAS($succ.prev, h',pred$)\label{remove2}\nl
CAS($h'.new, h_1, \GONE$)\label{gone}\nl
CAS($h_1.old, h', \DONE$)\label{done}\nl
\p\p {\bf end while} \nl
{\bf return} 
\end{code}
}
\caption{An Algorithm for Moving a Node to the Front of the List}
\label{mtfcode}
\end{figure}

If the set contains at most one item, an implementation is straightforward.
So, we assume that the set contains at least two items.
The items of the set are stored in a doubly linked list of nodes. A compare-and-swap object, $Head$, contains pointers to the first and second nodes in the list.
Every node in the list (excluding the first node in some intermediate configurations) 
contains a different item.

Each node has five fields:
\begin{itemize}
\item
$item$, a register
which contains an item and is never changed,
\item
$next$, a compare-and-swap object which points to the next node in the list or \NULL, if the node is the last node in the list,
\item
$prev$, a compare-and-swap object which points to the previous node in the list or \NULL, if the node is the first node in the list,
\item
$old$, a register which is initialized to \NULL\ when the node is newly created,
is only changed (by the process that created the node) from \NULL\ to point to another node containing the same  item, and,
if it is pointing to a node,
is only changed to \DONE,
\item
$new$, a compare-and-swap object which is initialized to \NULL\ when the node is newly created,
is only changed from \NULL to point to another node containing the same item,
and, if it is pointing to a node,
is only changed to \GONE.
\end{itemize}

We assume that, initially, the $old$ field of every node in the list
is \DONE, the $new$ field of every node in the list is \NULL, and
the list consists of exactly one node for each
item in the set, which contains that item.

Each process $p_i$ has a compare-and-swap object $A[i]$, which contains a pair.
Initially, $A[i] = (\NULL,\bot)$, which indicates that $p_i$ is not currently searching
for an item. 
If the second component of $A[i]$ is $e \neq \bot$, then $p_i$ is
searching for item $e$. In this case, the first component of $A[i]$ is a pointer
to a replacement node newly allocated by $p_i$ at the beginning of its search and which
acts as a unique identifier for the search.
When the second component of $A[i]$ is $\bot$, but its first component is not \NULL,
the first component is a pointer to a node that
was at the front of the list at some point after $p_i$ started its current search
and contains the item $p_i$ is searching for.
Process $p_i$  is the only process that changes the first component of $A[i]$
from \NULL or to \NULL and the only process that successfully changes $A[i]$
when its second component is $\bot$.

SEARCH($e$) is used by process $p_i$
to search for a node containing the item $e \neq \bot$.
It begins by setting $g$ to point to a newly allocated replacement node on line \lref{allocate},
setting its $item$ field to $e$ on line \lref{newitem},  and
setting all its other fields to \NULL on line \lref{newptrs}.
Next, $p_i$ announces $(g,e)$ in $A[i]$ on line \lref{announce}
and reads $Head$ on line \lref{gethead} to get pointers to the first and second nodes in the list. Then, it goes through the list, one node at a time, in order, comparing
its $item$ with $e$.

Suppose that, on line \lref{katfront}, process $p_i$ finds that the first node in the list 
(i.e.\ the node pointed to by the first pointer in $Head$) contains item $e$.
Then $p_i$ tries to reset
$A[i]$ to $(\NULL ,\bot)$ on line \lref{unannFront1} and returns a pointer 
to this node on line \lref{return1}.
Note that, if another process changed $A[i]$ between when $p_i$
announced that it was searching for $e$ and when $p_i$ tries to reset $A[i]$ to $(\NULL ,\bot)$,
then the second component of $A[i]$ will be $\bot$ instead of $e$. In this case,
$p_i$ performs a second CAS on line \lref{unannFront2} to reset $A[i]$ to $(\NULL ,\bot)$.

While searching the rest of the linked list for item $e$, process $p_i$ repeatedly
checks on line \lref{checkOther2} whether the second component of $A[i]$
has been set to $\bot$,
indicating that some other process has prepended a node containing $e$ to the list.
If so, the first component of $A[i]$ is a pointer to such a node, which $p_i$ returns
on line \lref{return4} after resetting $A[i]$ to $(\NULL,\bot)$ on line \lref{unannOther2}.
It does this check each time it has examined $\phi$ nodes,
for some integer constant $\phi \geq 1$.

If $p_i$ reaches the end of the list, it resets $A[i]$ to $(\NULL,\bot)$ on line \lref{unannEnd}.
On line \lref{return5},
it again checks whether the second component of $A[i]$ 
was $\bot$  and, if so, returns the pointer that was in the first component of $A[i]$.
If no other process informed $p_i$ before it reached the end of the list, then $e$ is not in the
list and $p_i$ returns {\sc not present}.

Now suppose that, on line \lref{kfoundlater}, $p_i$ finds a node $v$
containing $e$ which is not at the front of  the list.
On line \lref{checkOther1},
it also checks whether the second component of $A[i]$ 
is $\bot$ and, if so, on lines \lref{unannOther1} and line \lref{return2},
resets $A[i]$ to $(\NULL,\bot)$ 
and returns the pointer that was in the first component of $A[i]$.
Otherwise, on line \lref{setgold}, $p_{i}$ sets the $old$ field of its replacement node
to point to node $v$, indicating that $v$ is an old 
node containing $e$
which it is trying to replace.
Then $p_i$ tries to prepend its replacement node to the list.

To ensure that only one replacement for node $v$ is prepended to the list, 
$p_i$ first tries to change the compare-and-swap object $v.new$ from  \NULL\ to
point to its replacement node in line \lref{sethnew}.
Then it sets $g'$ to $v.new$ on line \lref{setgprime}.
If the CAS was successful, then $g'$ points to its replacement node.
If the CAS was unsuccessful, then either $g'$ points to some other node,
which is also a replacement for $v$,
or $g' = \GONE$, indicating that
node $v$ has already been removed from the list
(and a replacement for $v$ has already been prepended to the list).
If $g' \neq \GONE$, then $p_i$ tries to prepend the replacement node pointed to by $g'$
to the list
and remove $v$ from the list by calling MOVE-TO-FRONT($g'$).
In all cases, the first component of $A[i]$ now points to a replacement for node $v$
that has been prepended to the list.
Then,  on lines \lref{unannPostMTF} and \lref{return3}, $p_i$  resets $A[i]$ to $(\NULL,\bot)$ and returns the pointer that was in the first component of $A[i]$.

In MOVE-TO-FRONT($g'$), $p_i$ repeatedly performs the following steps until the replacement node pointed to by $g'$ has been prepended to the list.

It first reads $Head$ on line \lref{gethead}
to get pointers, $h_1$ and $h_2$, to the first two nodes in the list.
If the insertion of the first node is complete (i.e.\ $h_1.old = \DONE$), then,
on line \lref{tryprepend}, 
$p_i$ tries to prepend the node pointed to by $g'$
to the list by trying to change $Head$ from $(h_1,h_2)$ to $(w',h_1)$.
Note that other processes may also be trying concurrently to prepend the same node
or another node to the list.
In all cases, $p_i$ then helps complete the insertion of the
node at the front of the list.

To help complete the insertion of this node, $p_i$ first ensures that the first two nodes in the list point to one another
by trying to set the {\it next} field of the first node to point to the second on line \lref{setnext}
and trying to set the {\it prev} field of the second node to point to the first on line \lref{setprev}.
If either of these CAS operations is not successful, some other process did it first.

Then $p_i$ informs each other process $p_j$ 
that is currently looking for the item $e'$ in the node now at the front of the list.
Specifically, for every announcement $A[j]$ that contains $e'$, a CAS is performed
on line \lref{inform} to try to change it to $(h_1,\bot)$, provided $h_1$ still points to
the front of the list.
It is possible that $p_i$ could fall asleep for a long time between checking that
$h_1$ still points to the front of the list and performing the CAS.
In the meantime, it is possible that other nodes have been prepended to the 
list and $p_j$ has started another search for $e'$.
In this case, the CAS should fail. This is why the announcement $A[j]$
contains a unique identifier for the search
(which is a pointer to the replacement node allocated at the beginning of the search),
 in addition to the value being sought.

After this, the old copy of the new node at the front of the list
is deleted from the list, by changing
the {\it next} field of its predecessor and the {\it prev} field of its successor 
on lines \lref{remove1} and \lref{remove2}.
Note that if the old node was at the end of the list (i.e.~its $next$ pointer is $\NULL$),
the second of these CAS operations is not performed.
Finally, the {\it new} pointer in this old copy is set to \GONE\ on line \lref{gone} 
and the $old$ pointer in the newly inserted node is changed to $\DONE$ on line \lref{done}.

Note that when a node is removed from the list, a process that is traversing the list and is at
that node will be able to continue traversing the list as if the node had not been removed.
This is because the {\it next} field of the removed node continues to point to the node
that was its last successor.

The implementation can be made wait-free by using round-robin helping
when trying to prepend a node to the list. Specifically, in addition to 
the pointers to the first two nodes in the list, $Head$ contains 
a modulo $p$ counter, $priority$, which
indicates which process has priority for next prepending a node to the front of the list.
Each time $Head$ is modified, the counter is incremented.
Before trying to prepend the node $v'$ pointed to by $g'$, process $p_i$ checks $A[priority]$.
If its second component is $\bot$, its first component contains a pointer to another node, and the {\it old} field of that node is not yet \DONE, then 
$p_i$ tries to prepend this other node instead of $v'$.
Process $p_i$ repeats these steps at most $p$ times before node $v'$
is prepended.

\subsection{Correctness}
\label{correctness}

In this section, we prove that \DMTF is {\em linearizable}, which
is a standard definition of correctness for distributed data structures \cite{AW04}. 
This means that, for every execution, it is possible to assign a distinct linearization point
to every complete operation on the data structure and some subset of the incomplete
operations such that the following two properties hold.
First, the linearization point of each such operation occurs after it begins
and, if it is complete, before it ends. 
Second, 
the result of every completed operation in the original execution is the same as in the
corresponding {\em linearization}, the execution 
in which the linearized operations are performed sequentially in order of their linearization points.

We begin by
proving some observations and invariants about the $old$ and $new$ fields of nodes.

\begin{observation}
The $old$ field of a node only changes from \NULL to point to a node 
that has been in the list
and it only changes from pointing to a node to \DONE.
Once it is \DONE, it never changes.
\label{obs:old}
\end{observation}

\begin{proof}
The $old$ field of a node is only changed on lines \lref{setgold},
which changes it from \NULL to $h$,
and  \lref{done}, which changes it to \DONE.
Note that, before process $p_i$ performs line \lref{setgold},
its local variable $h$ is set to the second node in the list on line \lref{gethead}.
It is only updated on line \lref{next} by following $next$
pointers.
Since it is not \NULL, by the test on line \lref{htest}, $h$ points to a node that has been in the list.
It follows that the $old$ field of a node never points to a node that has not been in the list.
\end{proof}

\begin{observation}
The $new$ field of a node only changes from \NULL to point to a node 
whose $old$ field is not \NULL
and it only changes from pointing to a node to \GONE.
Once it is \GONE, it never changes.
\label{obs:new}
\end{observation}

\begin{proof}
The $new$ field of a node is only changed on lines  \lref{sethnew}, which changes it 
from \NULL to $g$, and \lref{gone}, which changes it to \GONE.
Note that, before process $p_i$ performs line \lref{sethnew},
it changes $g.old$ from \NULL on line \lref{setgold}, if it has not already been changed.
By Observation \ref{obs:old}, the $old$ field of a node never changes back to \NULL.
It follows that the $new$ field of a node never points to a node whose $old$ field is \NULL.
\end{proof}

\begin{observation}
When MOVE-TO-FRONT($g'$) is called, $g'$ points to a node whose $old$ field is 
not \NULL.
\label{obs:mtf}
\end{observation}

\begin{proof}
By line \lref{setgprime}, $g' = h.new$, where $h$ is a local variable of process $p_i$
that points to some node $u$.
On line \lref{sethnew}, the $new$ field of $u$ is changed from \NULL to point to a node,
if it has not already been changed. 
When $p_i$'s local variable $g'$ is set to $h.new$ on line \lref{setgprime},
Observation \ref{obs:new} implies that $h.new$ is either \GONE
or points to a node whose $old$ field is not \NULL.
By the test on line \lref{mtf}, $g' \neq \GONE$ when MOVE-TO-FRONT($g'$) is called, 
so $g'$ points to a node whose $old$ field is not \NULL.
\end{proof}

\begin{lemma}
\ 
\begin{description}
\item[(a)]
The $old$ field of a node that has not been in the list is either \NULL or points to a different node containing the same item.
\item[(b)]
The $old$ field of the first node in the list is either \DONE or points to a different node containing the same item.
\item[(c)]
The $old$ field of every node other than the first
in the list and every node that is no longer 
in the list is \DONE.
\end{description}
\label{inv:old}
\end{lemma}

\begin{proof}
by induction on the execution.
Initially, this is true since the $old$ field of every node in the list is \DONE.
When a node is newly allocated by a process $p_i$, its $old$ field is set to \NULL on line \lref{newptrs}, which ensures that
the claim continues to hold.
Its $old$ field is only changed on lines \lref{setgold} and  \lref{done}.

When the CAS on line \lref{setgold} is successfully performed, the $old$ field of the node $v$
that is pointed to by $p_i$'s local variable $g$ 
is changed from \NULL\ to $h$.
By Observation \ref{obs:old}, $h$ points to a node that has been in the list.
By the test on line \lref{kfoundlater}, $h.item = e$, which is the same as the $item$
field of $v$.
By the induction hypothesis,
$v$ has never been in the list.
Hence, immediately after line \lref{setgold},  the $old$ field of $v$ points to a different node containing the same item
and 
the claim continues to hold.

When the CAS on line \lref{done} is successfully performed,
it sets the $old$ field of the node $v$ to \DONE, where $v$ is the node 
pointed to by $p_i$'s local variable $h_1$. Note that $h_1$ was last set on line \lref{gethead2} to point to the first node in the list. Thus, immediately before line \lref{done} is performed, $v$ 
has been in the list, so by the induction hypothesis, its $old$ field is not \NULL.
Thus the claim continues to hold.

Now, consider what happens when the node $v$, pointed to by $p_i$'s local variable $g'$,
is prepended to the list  (by a successful CAS on line \lref{tryprepend}).
Immediately beforehand, the $old$ field of the first node in the list is \DONE, by lines \lref{gethead2} and \lref{getold}
and the test on line \lref{testold}.
By Observation \ref{obs:mtf}, the $old$ field of $v$ is not \NULL and,
by the test on line \lref{outerloop},
it is not \DONE.
Thus, it points to some node $u$.
By the induction hypothesis, $u$ and $v'$ are different nodes
containing the same item.
When $v$ is  prepended to the list, the first node becomes the second node
and $v$ becomes the first node, so the claim continues to hold.

Finally, suppose $p_i$ removes the node $v'$ pointed to by its local variable 
$h'$ from the list on line \lref{remove1}.
When $p_i$ performed line  \lref{gethead2}, it set its local variable $h_1$ to point to
the first node in the list. 
When $p_i$ performed line \lref{getold}, it set its local variable $h'$ equal to the $old$
field of this node, which was  either \DONE or pointed to a different node containing the same item, by the induction hypothesis.
By the test on line \lref{testold}, $h' \neq \DONE$.
Thus, $v'$ is not the first node in the list. By the induction hypothesis,
the $old$ field of $v'$ is \DONE. Hence, when $v'$ is removed from the list,
the claim continues to hold. 
\end{proof}

\begin{lemma}
\ 
\begin{description}
\item[(a)]
The $new$ field of a node that has not been in the list is \NULL.
\item[(b)]
The $new$ field of a node in the list is either \NULL
or points to a different node containing the same item.
\item[(c)]
The $new$ field of a node is \GONE only after it has been removed from the list.
\end{description}
\label{inv:new}
\end{lemma}

\begin{proof}
by induction on the execution.
Initially, this is true since the $new$ field of 
every node is \NULL.
When a node is newly allocated by a process $p_i$, its $new$ field is set to \NULL 
on line \lref{newptrs}, which ensures that the claim continues to hold.
Its $new$ field is only changed on lines \lref{sethnew} and \lref{gone}.

When the CAS on line \lref{sethnew} is successfully performed,
$h.new$ is changed from \NULL\ to $g$, which points to a node
that has not yet been in the list.
By line \lref{newitem} and the test on line \lref{kfoundlater}, $g.item  = e = h.item$.
Since $h$ points to a node in the list,
$h$ and $g$ point to different nodes containing the same item.
Thus
the claim continues to hold.

If the CAS on line \lref{gone} is successfully performed,
it sets the $new$ field of the node $u$ to \GONE, where $u$ is the node 
pointed to by $p_i$'s local variable $h'$. It occurs after $u$ has been removed from the
list on lines \lref{getprev}--\lref{remove2}.
Thus the claim continues to hold.
\end{proof}

\begin{lemma}
If the $new$ field of a node $u$ points to another node $v$, then the $old$ field of $v$
points to $u$.
\label{inv:newold}
\end{lemma}

\begin{proof} by induction on the execution.
Initially, this is true since the $new$ field of 
every node is \NULL and the $new$ field of every node is set to \NULL on line \lref{newptrs}
when it is allocated.

The $new$ field of a node is only changed to point to a node 
by a successful CAS on line \lref{sethnew}.
Let $u$ and $v$ denote the nodes to which $p_i$'s local variables $h$ and $g$ point
immediately before process $p_i$ performs line \lref{setgold}.
From the code, process $p_i$ allocated $v$ and  is the only process that knows about
this node. Thus, its CAS on line \lref{setgold} successfully changes
the $old$ field of node $v$ from \NULL to point to node $u$.
If  the CAS by process $p_i$ on line \lref{sethnew} is also successful,
then the $new$ field of node $u$ now points to $v$.

When the $old$ field of a node is changed from pointing to a node to \DONE on line \lref{done}, 
the $new$ field of the node it points to has already been changed to \GONE on line \lref{gone}.
\end{proof}

\begin{ignore}
if this CAS is not successful, then the $new$ field of $u$ has already been changed from \NULL.
By part (g) of the induction hypothesis, if it points to a node $v'$, then the $old$ field of $v'$
points to $u$.
Since the $new$ field of a node is never changed back to \NULL,
in both cases, immediately after $p_i$ performs line \lref{setgprime}, its local variable
$g'$ is either \GONE or points to a node whose $old$ field points to $u$.
Since the $old$ field of a node is never changed back to \NULL,
if and when $p_i$ calls MOVE-TO-FRONT($g'$) on line \lref{mtf},
$g'$ points to a node whose $old$ field is either \DONE or points to $u$.

Immediately before the node $v'$ pointed to by $g'$  is  prepended to the
list on line \lref{tryprepend}, 
the $old$ field of the first node in the list is \DONE, by lines \lref{gethead2} and \lref{getold}
and the test on  line \lref{testold}.
By the test on line \lref{outerloop}, the $old$ field of $v'$ is not \DONE,
so it points to some node $u$.
By the induction hypothesis, $u$ and $v'$ are different nodes
containing the same item.
Thus, when the node pointed to by $g'$ is  prepended to the
list on line \lref{tryprepend}, the claim continues to hold.
\end{ignore}

Next, we examine what happens when nodes are prepended and removed from the list.

\begin{ignore}
\begin{observation}
When the line on which a node is prepended to the list is performed,
the $old$ field of the node that was first in the list is \DONE.
When a node is being removed from the list,
the $old$ field of the first node in the list points to it.
\label{prerem}
\end{observation}

\begin{proof}
A node is prepended to the list on line \lref{tryprepend} and
is removed from the list on lines \lref{getprev}--\lref{remove2}.

Process $p_i$ sets its local variable $h_1$ to point to the first node, $v$, in the list on line \lref{gethead2}
and assigns $h_1.old$ to its local variable $h'$ on line \lref{getold}.
By the test on line \lref{testold}, when $p_i$ tries to prepend a node, $h' = \DONE$
and, by Observation \ref{obs:old}, the $old$ field of node $v$ does not subsequently change.
The CAS on line \lref{tryprepend} is successful only if the first two nodes in the list
have not changed since $p_i$ last performed line \lref{gethead2} and, hence,
 the $old$ field of the first node in the list is \DONE. 

Process $p_i$ tries to remove a node in the list only when $h' \neq \DONE$,
in which case, it tries to remove the node $u$ to which $h'$ points.
Note that the $old$ field of node $v$ is only changed on line \lref{done},
after $u$ has been removed by $p_i$ (or some other process) on line \lref{remove2}.
Thus, no other process can prepend a node to the list while a node is being removed 
from the list.
\end{proof}
\end{ignore}

\begin{observation}
Immediately before a node is prepended to the list,
the $old$ field of the first node in the list is \DONE.
\label{prerem1}
\end{observation}

\begin{proof}
A node is prepended to the list on line \lref{tryprepend}.
Process $p_i$ sets its local variable $h_1$ to point to the first node, $v$, in the list on line \lref{gethead2} and assigns $h_1.old$ to its local variable $h'$ on line \lref{getold}.
By the test on line \lref{testold}, when $p_i$ tries to prepend a node, $h' = \DONE$
and, by Observation \ref{obs:old}, the $old$ field of node $v$ does not subsequently change.
The CAS on line \lref{tryprepend} is successful only if the first two nodes in the list
have not changed since $p_i$ last performed line \lref{gethead2}. Hence,
immediately before the node is prepended, $v$ is the first node in the list and its
$old$ field is \DONE. 
\end{proof}

\begin{observation}
When a node is being removed from the list,
the $old$ field of the first node in the list points to it.
\label{prerem2}
\end{observation}

\begin{proof}
A node is removed from the list on lines \lref{getprev}--\lref{remove2}.
Process $p_i$ tries to remove a node in the list only when $h' \neq \DONE$,
in which case, it tries to remove the node $u$ to which $h'$ points.
Note that the $old$ field of node $v$ is only changed on line \lref{done},
after $u$ has been removed by $p_i$ (or some other process) on line \lref{remove2}.
Thus, no other process can prepend a node to the list while a node is being removed 
from the list.
\end{proof}

Thus, a node is not prepended to the list while another node is being removed
and only one node is removed from the list at a time.
By Lemma  \ref{inv:old}, the $old$ field of a node never points to itself.
Thus, when a node is removed, it is not the first node in the list.

\begin{lemma}
No node is prepended to the list more than once.
\label{prependedonce}
\end{lemma}

\begin{proof}
To obtain a contradiction, suppose there is a node that is prepended to the list more than once.
Let $v$ be the first node that is prepended to the list a second time and let $p_i$ 
the process that successfully performs the CAS on line \lref{tryprepend} to do this.
By Observation \ref{prerem1}, when it performs this line,
the $old$ field of the first node in the list is \DONE.
By Lemma \ref{inv:old}, the $old$ field of every other node that has been in the list is \DONE.
In particular, the $old$ field of $v$ is \DONE.
Thus, when $p_i$ performs the test on line \lref{tryprepend}, $g'.old = \DONE$ and 
$p_i$ does not perform the CAS. This is a contradiction.
\end{proof} 

\begin{lemma}
If more than one process tries to remove a node from the list, the effect is the same as if only
one process tries to remove it.
\label{removedonce}
\end{lemma}

\begin{proof}
When a node is allocated, its $next$ field is initialized to \NULL on line \lref{newptrs}
and, when the node is first in the list, its $next$ field is changed from \NULL to point
to the second node in the list on line \lref{setnext}. Subsequently, this field is only changed
on line \lref{remove1}, when its successor is removed from the list.
Likewise, the $prev$ field of a node is initialized to \NULL on line \lref{newptrs},
when the node is second in the list, its $next$ field is changed from \NULL to point to the 
first node of the list in line \lref{setprev}, and thereafter, is only changed
on line \lref{remove2}, when its predecessor is removed from the list.

By Observation \ref{prerem2}, only the node, $v$, pointed to by the $old$ field
of the first node in the list is removed.
All processes trying to remove node $v$ read its $prev$ field to get a pointer to its predecessor
(on line \lref{getprev}) and  
and its $next$ field to get a pointer to its successor (on line \lref{getnext}).
Note that once $v$ has been removed from the list, its $prev$ and $next$ fields do not change.
Processes use CAS to try to change $v$'s predecessor to point to $v$'s successor
on line \lref{remove1} and, if $v$ is not the last node in the list, to try to change $v$'s successor 
to point to $v$'s predecessor on \lref{remove2}.
The first such CAS operations are successful, but subsequent ones are not, since these nodes no longer point to $v$.
Finally, processes use CAS to try to change the $new$ field of $v$ to \GONE and 
the $old$ field of the first node in the list to \DONE. Since there are no steps
that change a field with value \GONE or \DONE, only the first such CAS operations are successful.
\end{proof}

\begin{lemma}
Each item of the set has a node in the list that contains it.
In every configuration, every node in the list whose $old$ field is \DONE contains a different item of the set.
\label{allitemspresent}
\end{lemma}

\begin{proof}
The proof is by induction on the execution.
Initially, the claim is true, since the $old$ field of every node in the list is \DONE and
the list consists of exactly one node for each item in the set.

By Observation \ref{prerem2}, when a node is removed from the list,
the old field of the first node in the list points to it.
By Lemma \ref{inv:old}, these two nodes contain the same item.
Thus, after a node is removed,
the claim continues to hold.

By Lemma \ref{inv:old}, the $old$ field of every node in the list is \DONE,
except possibly the first node.
The $old$ field of the first node in the list is changed to \DONE on line \lref{done}
after the node to which it pointed has been removed from the list on lines 
\lref{getprev}--\lref{remove2}.
Thus, when the $old$ field of a node is changed to \DONE, the claim continues to hold.

By the test on line \lref{tryprepend}, the $old$ field of a node being prepended
is not \DONE. The $old$ field of a node is only changed to \DONE on line \lref{done}
when the node is first in the list, so when the CAS on line \lref{done}
node is prepended, its $old$ field is still not 
\DONE.
Thus, prepending a node does not change the set of nodes whose $old$ field is \DONE
and the claim continues to hold.
\end{proof}

Next, we consider how the announcement array can change.

\begin{observation}
\ 
\begin{description}
\item[(a)]
When $A[i] = (\NULL, \bot)$,
process $p_i$ can change its first component to the node it allocated at the beginning of
its current search and its second component to the item it is searching for.
\item[(b)]
When the second component of $A[i]$ is an item, 
any process can change $A[i]$ so that its second component is $\bot$ and its first component points to a node that contains the item.
\item[(c)]
When the first component of $A[i]$ is not \NULL,
process $p_i$ can change it to $(\NULL,\bot)$.
\end{description}
No other changes to $A[i]$ are possible.
\label{Achange}
\end{observation}

\begin{proof}
From the code, process $p_i$ can change $A[i]$ from $(\NULL, \bot)$ on line
\lref{announce}, 
it can change $A[i]$ to $(\NULL, \bot)$  on lines \lref{unannFront1}, \lref{unannFront2}, \lref{unannOther1}, \lref{unannPostMTF}, \lref{unannOther2}, and \lref{unannEnd}, and it can change $A[j]$ for $j \neq i$ on line \lref{inform}.
The CAS on line \lref{inform} changes $A[j]$ to $(h_1,\bot)$
only if its second component was $e'$.
By line  \lref{h1item},
$h_1$ points to a node that contains the item $e'$.
\end{proof}

\begin{observation}
If the second component of $A[j]$ is $e'$ when a node containing $e'$ is prepended to the list, then it is changed to $\bot$ before the other node in the list containing $e'$ is removed from the list.
\label{informBeforeRemove}
\end{observation}

\begin{proof}
When a node is prepended to the list, its $old$ field is not \DONE, by Lemma \ref{inv:old}(a).
By Observation \ref{prerem1}, this field must be set to \DONE before another node is prepended. From the code, before the $old$ field of a node is set to \DONE, which
occurs on line \lref{done}, every announcement $A[j]$ whose second component is $e'$
is changed. This is because the CAS on line \lref{inform} is successful unless
$A[j]$ was changed since it was last read on line \lref{readannMTF}.
By Observation \ref{Achange}, if $A[j]$ was changed,
its second component was changed to $\bot$.
\end{proof}

\begin{lemma}
When process $p_i$ performs line \lref{readann2}, the second component of $A[i]$ is $\bot$.
\label{postMTF}
\end{lemma}

\begin{proof}
Let $v$ be the node pointed to by $p_i$'s local variable $h$.
By the test on line \lref{kfoundlater}, $h$ points to a node, $u$, containing $e$.
Note that $h$ was set to point to the second node in the list on line \lref{gethead} and it was only updated on line \lref{next} by following $next$ pointers.
Since nodes are only prepended to the list, $u$ was
in the list when $p_i$ performed line \lref{gethead}.
By the test on line \lref{katfront},  the node at the front of the list when
$p_i$ performed line \lref{gethead} does not contain $e$.
On line \lref{setgprime}, $p_i$'s local variable $g'$ is set to $u$'s $new$ field,
which, by the CAS on line \lref{sethnew}, is not \NULL.

If $g'$ is \GONE, then, by Lemma \ref{inv:new}, $u$ has been removed from the list.
Otherwise, by Lemmas \ref{inv:new} and \ref{inv:newold},
$u$'s $new$ field points to a different node $v$, whose
$old$ field points to $u$.
In this case,  $p_i$ calls MOVE-TO-FRONT($g'$), where $g'$ points to $v$.
By the test on line \lref{outerloop}, $v$'s $old$ field is \DONE when $p_i$ returns
from MOVE-TO-FRONT.
Thus, by Lemmas \ref{inv:newold} and \ref{obs:new}, $u$'s new field no longer points to $v$
and, hence is \GONE. Lemma \ref{inv:new} implies that  $u$ has been removed from the list.
In both cases, it follows from Observation \ref{informBeforeRemove} that
the second component of $A[i]$ has been changed to $\bot$.
By Observation \ref{Achange}, no other process changes $A[i]$ if its second component
is $\bot$.
Therefore, when process $p_i$ performs line \lref{readann2},
the second component of $A[i]$ is $\bot$.
\end{proof}

\begin{lemma}
When process $p_i$ performs line \lref{readann4}, the second component of $A[i]$ is $\bot$
if the item it is searching for is in the list.
\label{atEndButPresent}
\end{lemma}

\begin{proof}
Suppose item $e$  is in the set and 
$p_i$ performs line \lref{readann4} during an invocation of SEARCH($e$).
Lemma \ref{allitemspresent} implies that  there was a node $v$ containing $e$ in the list when $p_i$ performed line \lref{announce}.
When $p_i$ performed line \lref{gethead},
the node at the beginning of this list did not contain $e$, otherwise $p_i$ would have returned on line \lref{return1}.
Since $p_i$ reached the end of the list without finding a node containing $e$,
node $v$ must have been removed from the list
after $p_i$ announced its search on line  \lref{announce} and before it read
$A[i]$ on line \lref{readann4}.
It follows from Observation \ref{informBeforeRemove} that
the second component of $A[i]$ was changed to $\bot$.
By Observation \ref{Achange}, no other process changes $A[i]$ if its second component
is $\bot$.
From the code, $p_i$ does not change $A[i]$ back to
$(\NULL, \bot)$ before it performs line \lref{unannEnd}. 
Thus, when process $p_i$ performs line \lref{readann4},
the second component of $A[i]$ is $\bot$.
\end{proof}

\begin{lemma}
\label{quiescent}
When process $p_i$ is not performing SEARCH, $A[i] = (\NULL, \bot)$. 
When $A[i] = (\NULL, \bot)$, either $p_i$ is not performing SEARCH or has not yet performed line \lref{announce} during its current invocation of SEARCH.
\end{lemma}

\begin{proof}
The proof is by induction on the execution.
$A[i]$ is initially $(\NULL, \bot)$  and $p_i$ is not performing SEARCH.
We consider each step in the execution where $A[i]$ might change.

Before $p_i$ performs line \lref{announce} during an invocation of SEARCH($e$),
$A[i] =(\NULL, \bot)$, by the induction hypothesis.
Thus, the CAS on this line successfully changes $A[i]$ so that 
its second component is $e \neq \bot$.

On line \lref{inform}, $h_1$ is a pointer to the node at the front of the
list by line \lref{gethead2}
and $b$ is an item by the test on line \lref{inform}. Thus, this step does not change
$A[i]$ to or from $(\NULL, \bot)$.

The only other lines in which $A[i]$ is changed are
\lref{unannFront1}, \lref{unannFront2}, \lref{unannOther1}, \lref{unannPostMTF}, \lref{unannOther2}, and \lref{unannEnd}.
In all these lines, $p_i$ uses a CAS to try to change $A[i]$ to $(\NULL, \bot)$.
With the exception of the CAS on line \lref{unannFront1},
process $p_i$ returns immediately following each of those lines.
However, if the CAS on line \lref{unannFront1} is successful,
then the CAS on line \lref{unannFront2} is not performed,
so $p_i$ also returns immediately in this case.
Thus, $p_i$ returns from SEARCH immediately after $A[i]$ is set to $(\NULL, \bot)$.

It remains to show that when $p_i$ returns from SEARCH, $A[i] = (\NULL, \bot)$.
Note that $p_i$ changes $A[i]$ from $(\NULL, \bot)$
exactly once during each invocation of SEARCH, when it performs  line \lref{announce}.
By Observation \ref{Achange}, no other process changes $A[i]$ from (or to)
$(\NULL, \bot)$.
Thus, it suffices to show that $p_i$ changes $A[i]$ to $(\NULL, \bot)$
between when it performs line \lref{announce} and it returns from SEARCH.

First suppose that $p_i$ returns from SEARCH on line \lref{return1}.
Prior to this, it performs the CAS on line \lref{unannFront1}.
If this CAS was successful, $A[i]$ is changed to  $(\NULL, \bot)$.
If this CAS was unsuccessful, $A[i]$ was changed by another process since $p_i$ set it to 
$(w,e)$ on line \lref{announce}.
By Observation \ref{Achange}, the second component of $A[i]$ was changed to $\bot$
and $A[i]$ does not change again until $p_i$ changes it.
Hence $p_i$ performs the CAS on line \lref{unannFront2},
which successfully changes $A[i]$ to $(\NULL, \bot)$.

The only other lines on which $p_i$ returns from SEARCH are
\lref{return2}, \lref{return3}, \lref{return4}, and \lref{return5}.
Immediately prior to performing any of these lines, $p_i$ reads $A[i]$ and then performs a CAS.
If no other process changes $A[i]$ between these two steps, the CAS
sets $A[i]$ to $(\NULL, \bot)$. 
By Observation \ref{Achange}, no other process changes $A[i]$ if its second component
is $\bot$.

Immediately prior to when $p_i$ returns from SEARCH on line \lref{return2}, 
it performs the CAS on line \lref{unannOther1}.
By the test on line \lref{checkOther1}, the second component of
$A[i]$ was $\bot$ when $p_i$ read $A[i]$ on line \lref{readann1}.
Thus the CAS successfully changes $A[i]$ to $(\NULL, \bot)$.

Similarly, immediately prior to when $p_i$  returns from SEARCH on line \lref{return4}, 
it performs the CAS on line \lref{unannOther2}.
By the test on line \lref{checkOther2}, the second component of
$A[i]$ was $\bot$ when $p_i$ read $A[i]$ on line \lref{readann3}, so
the CAS successfully changes $A[i]$ to $(\NULL, \bot)$.

Now, suppose $p_i$ returns from SEARCH($e$) on line \lref{return3}.
By Lemma \ref{postMTF}, when $p_i$ performed line \lref{readann2}, the second component of $A[i]$ was $\bot$.
By Observation \ref{Achange}, no other process changes $A[i]$ if its second component
is $\bot$.
Hence, after the CAS on line \lref{unannPostMTF}, $A[i] = (\NULL, \bot)$.

Finally, suppose that $p_i$ returns from SEARCH on line  \lref{return5}.
When $p_i$ performs the CAS on line \lref{unannEnd}, it has reached the end of the list without finding a node with item $e$.
If $e$ is not in the set, then Observation \ref{Achange} implies that 
no other process can change $A[i]$ from $(w,e)$ to anything else.
If $e$ is  in the set, then, by Lemma \ref{atEndButPresent}, 
the second component of $A[i]$ was $\bot$
when process $p_i$ performed line \lref{readann4}.
By Observation \ref{Achange}, $A[i]$ was not changed by any other process between when
$p_i$ performed lines \lref{readann4} and \lref{unannEnd}.
Thus, in both cases,
the  CAS on line \lref{unannEnd} successfully changes $A[i]$ to $(\NULL, \bot)$.
\end{proof}

\begin{ignore}
When process $p_i$ performs line \lref{readann4}, then, by Lemma \ref{allitemspresent}, there was a node $v$ containing $e$ in the list when $p_i$ performed line \lref{announce}.
At that point, 
the node at the beginning of this list did not contain $e$, otherwise $p_i$ would have returned on line \lref{return1}.
Hence, $v$ must have been removed from the list
after $p_i$ announced its search on line  \lref{announce} and before it read
$A[i]$ on line \lref{readann4}.
Before $v$ was removed from the list, a new node with item $e$
was prepended to the list and the second component of every announcement $A[j]$ containing $e$ was changed to $\bot$.
Since all this occurred after $p_i$  set $A[i]$ to $(w,e)$ on line \lref{announce}, 
the second component of $A[i]$ was changed by another process after $p_i$
performed line \lref{announce} and before $p_i$ read
$A[i]$ on line \lref{readann4}. Thus, by 
Observation \ref{Achange}, $A[i]$ was not changed by any other process between when
$p_i$ performed lines \lref{readann4} and \lref{unannEnd}.
In both cases,
the  CAS on line \lref{unannEnd} successfully changes $A[i]$ to $(\NULL, \bot)$.
\end{ignore}

\begin{lemma}
Suppose a process invokes SEARCH($e$), where $e$ is not an item in the set.
If the process 
does not crash,
then it
returns \textsc{not present}.
\end{lemma}

\begin{proof}
Suppose $p_i$ invokes SEARCH($e$), where $e$ is not an item of the set and, hence,
is not contained in any node in the list.
While $p_i$ is performing SEARCH($e$), its tests on lines \lref{katfront} and \lref{kfoundlater} are never successful.
Thus $p_i$ does not return on line \lref{return1}, \lref{return2}, or  \lref{return3}.
Moreover, $p_i$ does not call MOVE-TO-FRONT.
By Lemma \ref{quiescent}, $A[i] \neq (\NULL,\bot)$ between the step in
which $p_i$ sets the second component of $A[i]$ to $e$
on line \lref{announce} and the step in which it returns.
By Observation \ref{Achange},  no other process changes $A[i]$ when its second component is $e$. Hence, $p_i$'s test on line \lref{checkOther2} is never successful while 
$p_i$ is performing SEARCH($e$) and, so, $p_i$ does not return on line  \lref{return4}.

Therefore, in each iteration of the loop,
$p_i$ updates $h$ on line \lref{next} to point to the next node in the list.
Since $h$ was set to point to the second node in the list on line \lref{gethead}
and nodes are only prepended to the list, eventually the end of the list is reached
and $h = \NULL$.
The second component of $A[i]$ is still $e$ when $p_i$ performs line \lref{readann4}.
Thus, SEARCH($e$)  returns \textsc{not present} on line \lref{return5}.
\end{proof}

If $e$ is not an item in the set, then SEARCH($e$) does not modify the list. 
Hence, this operation can be linearized
at any point during its operation interval,
for example, when it returns.
(Note that this case is not an option in the problem on which we perform a competitive analysis.
Both this algorithm and \OPT need to search the entire list for such items,
so it is not a case where it performs poorly in comparison to \OPT.)

\begin{lemma}
If process $p_i$ returns from SEARCH($e$) on line \lref{return1},
then $p_i$ returns a pointer to a node containing $e$, which was at the front of the list
when $p_i$ read $Head$ on line \lref{gethead}.
\end{lemma}

\begin{proof}
Suppose process $p_i$ returns from SEARCH($e$) on line \lref{return1}. Then, as a result of performing line \lref{gethead},
$h_1$ is set to point to the first node in the list. By the test on line \lref{katfront}, this node contains $e$. Then it returns $h_1$ on line \lref{return1}.
\end{proof}

In this case, the instance of SEARCH($e$) by process $p_i$ is linearized when $p_i$ read $Head$ on line \lref{gethead}.

\begin{ignore}
\begin{lemma}
When $A[j]$ is changed
so that its second component is $\bot$ and its first component points to a node, 
the node is first in the list.
FALSE:
Suppose p1 and p2 are performing MOVE-TO-FRONT and are both informing other
processes about node v which is at the front of the list and contains item 3'.
p1 is speedy and finishes the announcement loop before p2 starts.
p0 now calls SEARCH(e') and announces this in A[0].
Process p2 reads A[0] and then checks that v.old is not DONE.
Then p1 sets the old field of v to DONE and then prepends a new node to the list.
Finally p2 changes the second component of A[0] to \bot and its first component to point to v,
which is no longer at the front of the list.
\end{lemma}
\end{ignore}

\begin{lemma}
If $A[j]$ is changed
to $(h_1,\bot)$, then the node to which $h_1$ points
was at the front of the list at some time since $p_j$
announced its current search.
\label{atfront}
\end{lemma}

\begin{proof}
Suppose $A[j]$ is changed from $(a,b)$ to $(h_1,\bot)$ by processor $p_i$.
By Observation \ref{Achange}, $b$ is the item $p_j$ is currently searching for,
$a$ is a pointer to the
node that $p_j$ allocated at the beginning of this search,
and $A[j]$ has not changed since $p_j$ announced its current search on 
line \lref{announce}.
By line  \lref{h1item},
$p_i$'s local variable
$h_1$ points to a node $v$ that contains the item $e'$.
When $p_i$ last performed line \lref{gethead2}, $v$ was at the front of the list.
After checking that $b = e'$,
$p_i$ checks that $h_1.old \neq \DONE$, which, by Lemma \ref{inv:old},
implies that  $v$  is still at the front of the list.
\end{proof}

\begin{lemma}
If $e$ is in the list and process $p_i$ returns from SEARCH($e$) on line \lref{return5},
then $p_i$ returns a pointer to a node containing $e$, which was at the front of the list
at some time since $p_i$ announced this search on line \lref{announce}.
\label{return:end}
\end{lemma}

\begin{proof}
By Lemma \ref{atEndButPresent}, 
the second component of $A[i]$ was $\bot$
when process $p_i$ performed line \lref{readann4}.
From the code, $p_i$ has not tried to change $A[i]$ to $(\NULL,\bot)$ since it changed
it from $(\NULL,\bot)$ on line \lref{announce}.
By Observation \ref{Achange}, it follows that $A[i] = (h_1,\bot)$, where $h_1$ points
to a node, $v$, that contains $e$.
By the test on line \lref{return5}, $p_i$ returns $h_1$.
By Lemma \ref{atfront}, when $A[i]$ was changed to $(h_1,\bot)$,
$v$ was at the front of the list.
Thus, $p_i$ returns a pointer to a node containing $e$, which was at the front of the list
at some time since $p_i$ announced this search on line \lref{announce}.
\end{proof}

\begin{lemma}
If process $p_i$ returns from SEARCH($e$) on line \lref{return2}, \lref{return3}, or \lref{return4},
then $p_i$ returns a pointer to a node containing $e$, which was at the front of the list
at some time since $p_i$ announced this search on line \lref{announce}.
\label{return:middle}
\end{lemma}

\begin{proof}
If $p_i$ returns from SEARCH($e$) on line \lref{return2}, \lref{return3}, or \lref{return4},
then it has performed line \lref{unannOther1},  \lref{unannPostMTF}, or \lref{unannOther2}.
If $p_i$ performs the CAS on line \lref{unannOther1}, then
by the test on line \lref{checkOther1}, the second component of
$A[i]$ was $\bot$ when $p_i$ read $A[i]$ on line \lref{readann1}.
If $p_i$ performs the CAS on line \lref{unannOther2}, then
by the test on line \lref{checkOther2}, the second component of
$A[i]$ was $\bot$ when $p_i$ read $A[i]$ on line \lref{readann3}.
By Lemma \ref{postMTF}, 
when process $p_i$ performs line \lref{readann2}, the second component of $A[i]$ is $\bot$. 
By Observation \ref{Achange}, no other process changes $A[i]$ if its second component
is $\bot$.
Hence, immediately before  process $p_i$ performs the CAS on line \lref{unannOther1}, \lref{unannPostMTF}, or \lref{unannOther2}, the second component of $A[i]$ is $\bot$.
From the code, $p_i$ has not tried to change $A[i]$ to $(\NULL,\bot)$ since it changed
it from $(\NULL,\bot)$ on line \lref{announce}.
By Observation \ref{Achange}, it follows that $A[i] = (h_1,\bot)$, where $h_1$ points
to a node, $v$, that contains $e$.
By Lemma \ref{atfront}, when $A[i]$ was changed to $(h_1,\bot)$,
$v$ was at the front of the list.
Thus, $p_i$ returns a pointer to a node containing $e$, which was at the front of the list
at some time since $p_i$ announced this search on line \lref{announce}.
\end{proof}

In these remaining cases, the instance of SEARCH($e$) by process $p_i$,
which returns a pointer to a node $v$ containing $e$, can be 
linearized at any point which is after $p_i$ announces this search on line \lref{announce},
before $p_i$ returns from the search, and
at which $v$ was at the front of the list.

Thus, we have shown that the algorithm, \DMTF, presented in Figures 1 and 2 is correct.

\begin{theorem}
When a SEARCH($e$) operation is linearized, a node containing $e$ is at the front of the list
and the operation returns a pointer to this node, if $e$ is in the set.
Otherwise, the operation returns \textsc{not present}.
\end{theorem}

\section{Competitive Analysis of \DMTF -- Fully Adversarial}
\label{subsection-comp-analysis}

The analysis at the item access level counts the number of accesses
performed per request. It ignores all costs not directly associated with
the search.

In \DMTF,
each item accessed in the list involves a constant number of shared memory operations.
Furthermore, a process does an additional check of its announcement array
once every $\phi$ nodes.
Thus, the cost of a request is an $O(1+1/\phi)$ factor more than the number of item
accesses performed to handle the request.
 
There is at most one move-to-front operation per request.
Each successfully completed move-to-front would take $\Theta(p)$ steps if
it was performed sequentially: a constant number of updates of fields
in nodes and $\Theta(p)$ steps for informing the other processes.
However, in a distributed execution, it is possible that all $p$ processes
help perform the move-to-front. 
Thus, each request contributes $O(p^2)$ to the cost.
There is $O(\phi)$ extra cost per request for the nodes a process accesses in the list
after it has been informed that a node containing the item it is searching for has been
moved to the front of the list.
Finally, because the item in the node at the front of the list can occur in another
node of the list, there is an $O(1)$ additional cost.

Since \OPT's cost is at least 1 for each request, 
it follows from Theorem~\ref{theorem-item-level}
that, for the fully adversarial scheduler, if $\phi \in O(p^2)$, then
  \[\DMTF(I) \leq O(p^2)\OPT(I) + O(1),\]
  where  the additive constant depends only on  $p$, $\ell$, and $\phi$.
The lower bound of Theorem~\ref{theorem-ratio-lower-bound} shows that
any distributed algorithm (even one which treats the requests in an
optimal manner sequentially), must have
 \[\DMTF(I) \geq (2p^2-p)\OPT(I) - O(1),\]
if request sequences can be merged arbitrarily.
Therefore,  the cost of merging of the sequences
dominates the other costs and
$\DMTF(I) = \Theta(p^2)\OPT(I) + O(1)$.

\section{A Linearization-Based Analysis}
\label{sec:LinBased}

At the operation level, since the comparison is made on the same sequences,
the classic result from online algorithms gives the ratio $2-\frac{2}{\ell+1}$~\cite{I91}.
We proceed to the item access level.
 
Consider any execution $\epsilon$ of \DMTF on the request sequences $\Ip$.
Let  $\DMTF(\epsilon)$ denote the cost of $\epsilon$ at the item access level.
Let  $\BL{\Ip}$ denote the sequence of requests served in some linearization of $\epsilon$
and let $\OPT(\BL{\Ip})$ denote the cost of an optimal sequential execution on $\BL{\Ip}$.

Note that \MTF's cost on $\BL{\Ip}$ could be much less than $\DMTF(\epsilon)$:
When an item $x$ is requested $k$ times in a
row in $\BL{\Ip}$, \DMTF could have had up to $\min\{ p,k \}$ processes
concurrently searching for $x$ in the list. Those processes
could each incur cost equal to the index, $i$, that $x$ had in the list.
In the case where the move-to-front occurs before the other
processes begin searching for $x$, the first process incurs
cost $i$ and the remainder incur cost 1.

To compare $\DMTF(\epsilon)$ and $\OPT(\BL{\Ip})$,
we use the \emph{list factoring} and
\emph{phase partitioning techniques} as discussed in \cite{BE98}, using
the partial cost model. With list factoring, each distinct pair of items, $x$
and $y$, in the original list, $L$, is considered separately.
They are considered in a list, $L_{x,y}$, containing only these two items (in the same order as in $L$ at any point in time)
and with the subsequence of requests,  $\BL{\Ip}_{x,y}$, to these items that
occur in  $\BL{\Ip}$ (also in the same order).

The \emph{pairwise property} says that 
the items $x$ and $y$ are in the same order with respect to each other in $L$
at every point during the request sequence $\BL{\Ip}$ as they are in $L_{x,y}$ at corresponding point in the request sequence $\BL{\Ip}_{x,y}$.
The pairwise property holds for \MTF because, whenever a request to an item in the list
is served, the item is moved to the front of the list.
It holds for \DMTF because, in addition,
each request to an item in the list is linearized when the item is at the front of the list.
In the \emph{partial cost model}, the cost of every search is one less than in
the full cost model; only unsuccessful item inspections are counted.
Cost independence means that the algorithm makes decisions regardless of
the cost it pays, i.e., it behaves the same under all cost models. \DMTF
is cost independent.

The list factoring and phase partitioning techniques~\cite{BM85,I91,T93,AvSW95,A98} have become standard in studying the list accessing problem. The following results
are well known (see \cite{BE98}). 
Consider the execution $\epsilon_{x,y}$ obtained from $\epsilon$
by removing all steps by processes while they are not performing SEARCH($x$)
or SEARCH($y$) and all accesses to nodes except those containing $x$ or $y$.
Note that  \DMTF does not use 
paid exchanges and has the pairwise property. Thus, if we can show that,
for every sequence and for every pair of items, $x$ and $y$, the
partial cost of $\epsilon_{x,y}$ is at most $c$ times the partial
cost of \OPT on $\BL{\Ip}_{x,y}$, then the partial cost of 
$\epsilon$
is at most $c$ times the partial cost of \OPT on $\BL{\Ip}$. Since \DMTF
is cost independent, it follows that $\DMTF(\epsilon) \leq c \OPT(\BL{\Ip})$.

Recall that, with item access level interleaving, we only consider the costs of 
accessing nodes while
searching through the list, with total cost $i$ for an item in location $i$ of the list.

\begin{theorem}
\label{linearization-item-access}
With respect to item access level interleaving in the linearization-based model, 
for any execution, $\epsilon$, of \DMTF on the request sequences $\Ip$,
$$\DMTF(\epsilon) \leq (p + 1) \OPT(\BL{\Ip}) +O(1),$$
where the constant depends only on $p$ and $\ell$.
\end{theorem}
\begin{proof}
We partition
$\BL{\Ip}_{x,y}$
into \textit{phases} which are defined inductively 
as follows. Assume that, for some $t \geq 1$, we have defined phases up until, but not including,  the $t$'th request and the relative order of the two items in \OPT's
list is $x,y$ before the $t$'th request. Then the next phase is of 
\emph{type 1} and is of one of the following forms, where 
$j \geq 0$ and $k \geq 1$: 
\begin{equation*}
(a)~ yyy^j \ \ \  (b)~ (yx)^kyyy^j  \ \ \ \mbox{ and  } (c)~ (yx)^kxx^j.
\end{equation*}
Phases continue until the next request to a different item
(to $x$ in forms (a) and (b), and to $y$ in form (c)). There may
be one incomplete phase for each pair of items, but this only adds
a constant to the costs. Note that this phase partitioning is different
than the most common type, since we do not start a new phase after only
two identical requests, but wait until we get a different request.
This is necessary for the analysis, since processes may be searching
for the same item at the same time.

In case the relative order of the items is $y,x$ before the $t$'th request, 
the phase has \emph{type 2} and its form is exactly the same as above with $x$ and 
$y$ interchanged.

Table~\ref{phases} shows the costs incurred by the two algorithms,
and the worst case ratio, for each of the three forms of type 1. Note
that the same results hold for type 2 forms.

\begin{table}[ht]
\begin{center}
\begin{tabular}{c@{\hspace{2em}}|@{\hspace{2em}}c@{\hspace{2em}}c@{\hspace{2em}}|@{\hspace{2em}}c}
\toprule
\multicolumn{1}{@{\hspace{-2em}}c}{Phase} &
\multicolumn{1}{@{\hspace{-1.7em}}c}{\DMTF} &
\multicolumn{1}{@{\hspace{-3.3em}}c}{\OPT} &
\multicolumn{1}{c@{\hspace{2em}}}{Ratio}                     \\
\midrule
$yyy^j$	       & $\leq p$      & $1$   & $\leq p$                  \\[1ex]
$(yx)^{k}yyy^j$ & $\leq 2k+p$   & $k+1$ & $\leq 2 +\frac{p-2}{k+1}$ \\[1ex]
$(yx)^{k}xx^j$  & $\leq 2k+p-1$ & $k$   & $\leq 2+ \frac{p-1}{k}$   \\	
\bottomrule
\end{tabular}
	\caption{The costs of \DMTF and \OPT under the 
partial cost model for a phase of type 1 
(i.e., the initial ordering of items is $x,y$) and the maximum ratio
of these costs.}
	\label{phases}
\end{center}
\end{table}

The column for \DMTF holds because, in phases of form (a), at most
$p$ processes would find $y$ at index $2$ (and thus have cost $1$
in the partial cost model); in phases of form (b) and (c), \DMTF
has partial cost $1$ for each of the alternating occurrences of
$x$ and $y$, plus at most partial cost $p$ for the requests to the
same item at the end.

For form (a), \OPT moves $y$ to the front immediately, so it only
has cost $1$. For forms (b) and (c), \OPT does not do any moves while
the $x$ and $y$ are alternating, but in form (b), it moves the first
of the $y$'s after the alternation to the front. Thus, it has cost
$k+1$ for form (b) and $k$ for form (c).

The maximum ratio of \DMTF's to \OPT's performance is, thus, bounded
by $\max\{ p,2+\frac{p-1}{k}\}$. As mentioned above, since \DMTF
is cost independent, this bound also holds in the full cost model.
The term $p$ in $\max\{ p,2+\frac{p-1}{k}\}$
is
the dominating
term as long as $p\geq 3$ and $k\geq 2$. If $p=2$, the larger term
is $2+\frac{1}{k}\leq 3$. If $p=1$, then the result
is the standard $2$-competitiveness of \MTF. If $k=1$, then the
result is $2=p+1$. Thus, one concludes that, for any 
execution $\epsilon$, $\DMTF(\epsilon)\leq (p+1)\OPT(\BL{\Ip})$.
\end{proof}

At the level of the actual algorithm, \DMTF, the analysis from the
fully adversarial model shows that
the cost of each search
only increases by an $O(1+1/\phi)$ factor and an additive $O(p^2 + \phi)$ term.
Thus,
when $\phi \in O(p)$ and \OPT's average cost per request is
$\Omega(p)$, \DMTF is $\Theta(p)$-competitive.

\section{Concluding Remarks}
The List Accessing problem is the first self-adjusting data structure
problem considered in a distributed setting, where the processes each 
have their own request sequence, and a competitive analysis is performed.
It seems reasonable to assume that the concerns about the power of
the scheduler, which have arisen for the list accessing problem,
would also arise for other distributed data structure problems, 
where the individual processes have their own request sequences.

We have presented two different models for performing competitive analysis
in a distributed setting. More online problems in a distributed setting
should be investigated to determine how one best assesses the quality
of such algorithms. In that context, it is interesting to know what the
effect of the scheduler is.

\bibliography{refs}
\bibliographystyle{plain}

\end{document}